\documentclass[12pt,english,letterpaper]{article}
\pdfoutput=1

\usepackage{natbib} 
\bibpunct{(}{)}{;}{a}{,}{,}

\usepackage[english]{babel}

\usepackage{amsmath}
\usepackage{bm}
\usepackage{mathptmx}
\usepackage{amsthm}

\usepackage[T1]{fontenc} 
\usepackage[utf8]{inputenc} 

\usepackage{amssymb}
\usepackage{mathtools}

\usepackage{graphicx}

\usepackage{microtype} 
\usepackage{rotating} 
\usepackage{booktabs}  
\usepackage{array}     
\usepackage{multirow} 
\usepackage{enumerate}

\usepackage{titlesec}
\titleformat{\section}[block]{}{\thesection.}{5pt}{\bf\normalsize}
\titleformat{\subsection}[block]{}{\thesubsection.}{5pt}{\normalsize\emph}
\titleformat{\subsubsection}[block]{}{\thesubsubsection.}{5pt}{}

\DeclareMathOperator{\dd}{\mathrm{d}} 
\newcommand{\1}{\mathbf{1}}

\newcommand{\dev}{{\mathrm{dev}}}
\newcommand{\Lgf}{{\mathrm{L}}}
\newcommand{\Ngf}{{\mathrm{N}}}
\newcommand{\upp}{{\mathrm{upp}}}
\newcommand{\low}{{\mathrm{low}}}
\newcommand{\lib}{{\mathrm{lib}}}
\newcommand{\var}{{\mathrm{var}}}
\newcommand{\cons}{{\mathrm{cons}}}
\newcommand{\env}{{\mathrm{env}}}
\newcommand{\fin}{{\mathrm{fin}}}

\newcommand{\Lvec}{{\mathbf{L}}}
\newcommand{\Nvec}{{\mathbf{N}}}
\newcommand{\thetavec}{\text{$\boldsymbol\theta$}}
\newcommand{\inhom}{{\mathrm{inhom}}}

\DeclareMathOperator{\E}{E}
\newcommand{\x}{{\mathbf{x}}}

\usepackage[width=6in]{geometry} 

\newtheorem{theorem}{Theorem}[section]
\newtheorem{lemma}{Lemma}[section]
\newtheorem{proposition}[theorem]{Proposition}

\title{Global envelope tests for spatial processes}

\author{Mari Myllym{\"a}ki\footnote{Aalto University, Espoo, Finland and Natural Resources Institute Finland (Luke), Vantaa, Finland (email: mari.myllymaki@luke.fi)},
Tom\' a\v s Mrkvi\v cka\footnote{Faculty of Economics, University of South Bohemia, \v{C}esk\'e Bud\v{e}jovice, Czech Republic (email: mrkvicka.toma@gmail.com)},
Pavel Grabarnik\footnote{Institute of Physico-Chemical and Biological Problems in Soil Science of Russian Academy of Sciences, Pushchino, Russia (email: gpya@rambler.ru)},
Henri Seijo\footnote{Department of Computer Science, Aalto University, Espoo, Finland (email: henri.seijo@aalto.fi)},
Ute Hahn\footnote{Department of Mathematics, Aarhus University, Aarhus, Denmark (email: ute@math.au.dk)}}
\date{}

\begin{document}

\maketitle

\begin{abstract}

Envelope tests are a popular tool in spatial statistics, where they are used in goodness-of-fit testing.
These tests graphically compare an empirical function
$T(r)$ with its simulated counterparts from the null model.
However, the type I error probability $\alpha$ is
conventionally controlled for a fixed distance $r$ only, whereas the functions are
inspected on an interval of distances $I$.
In this study, we propose two approaches related to Barnard's Monte Carlo test for building global envelope tests on $I$:
(1) ordering the empirical and simulated functions based on their $r$-wise ranks
among each other, and
(2) the construction of envelopes for a deviation test.
These new tests allow the {\em a priori} selection of the global $\alpha$ and they yield $p$-values.
We illustrate these tests using simulated and real point pattern data.

\end{abstract}

\section{Introduction}\label{sec:intro}

Spatial patterns are highly complex statistical objects, so testing their null hypotheses
is non-trivial. The information contained in
the patterns must be summarised in a sensible manner to reflect the relevant features of the phenomenon under study. Typically, a null model is deemed appropriate if it explains the spatial interaction structure in the observed pattern.
Popular statistics used to capture spatial interaction comprise functions $T(r)$ that depend on inter-point distances $r$. For example, Ripley's $K$-function and the closely related $L$-function
are employed in the point process case, and the spherical contact distribution function (or
empty space function) in the random set case \citep[e.g., see][]{Cressie1993, IllianEtal2008, ChiuSKM2013, Diggle2013}. In order to test whether the null model is appropriate, we compare the empirical function $T_{\text{obs}}(r)$ estimated for the observed pattern with the distribution of $T(r)$ under the null model.
However, the theoretical distribution of $T(r)$ is unknown even in the simplest cases, so we have to resort to Monte Carlo methods for statistical inference.
In this study, we consider the popular \emph{envelope tests}, which display test results in a graphical manner.

The first step in a Monte Carlo envelope test is to generate $s$ realizations of the null model and to calculate $T_i(r)$, $i=2,\dots,s+1$, in each case for distances $r$ on a certain interval $I$.
The next step is to define critical bounds for the deviation of the empirical function $T_{\text{obs}}(r)$ ($\equiv T_1(r)$ in the following) from the expected value for each chosen $r$ and to display the results graphically.
It is often overlooked that rigorous statistical inference must be performed for all of the chosen distances $r$ simultaneously. This multiple testing problem, which was ignored in practice for a long time, has recently attracted much attention \citep{LoosmoreFord2006, GrabarnikEtAl2011, BaddeleyEtal2014}.

In spatial statistics, the envelope method was introduced
in two influential independent studies by \citet{Ripley1977} and \citet{BesagDiggle1977}.
The idea of this method is to compare the empirical function to an envelope given
by the $k$-th (often $k=1$) smallest and largest values of the simulated functions
$T_i(r)$ for each $r\in I$.
This method indicates the distances $r$ at which the behaviour of the empirical $T(r)$
may lead to the rejection of the null hypothesis.
This information is important for determining why the data
contradict the tested hypothesis, so this technique has become popular in many applied areas,
such as spatial ecology \citep[e.g.][]{WiegandMoloney2014},
biological sciences \citep[e.g.][]{SchladitzEtal2003, MattfeldtEtal2006, MattfeldtEtal2009, WestonEtal2012},
geography \citep[e.g.][]{CliffOrd1983}, and
astronomy \citep[e.g.][]{VioEtal2007, MartinezEtal2010}
\citep[also see][]{CaseStudies2006}.
Although the envelope method proposed by \citet{Ripley1977} should be interpreted
only as a diagnostic plot, as described in a recent review of the testing of spatial point
patterns by \citet{BaddeleyEtal2014}, it is desirable to use it as a proper statistical test.

In general, we understand an \emph{envelope} as a band bounded by the functions $T_\low(r)$ and $T_\upp(r)$ on an interval $I$.
A \emph{global envelope test} is a statistical test that rejects the null hypothesis if the observed function $T_1$ is not completely inside the envelope, i.e.,
\begin{equation}\label{eq:general envtest}
\varphi_\env(T_1)=\1(\exists\; r\in I:\ T_1(r) \notin (T_\low(r), T_\upp(r))).
\end{equation}
In this study, we consider two approaches for determining the bounds such that
the test \eqref{eq:general envtest} has a controlled overall level (global type I error probability)
for the {\em a priori} selected number of simulations $s$.
In the first approach, the bounds $T_\low(r)$ and $T_\upp(r)$ are constructed directly from the functions $T_i(r)$.
This leads to the {\em rank envelope test} (Section \ref{sec:rank_test}).
In the second approach, the bounds are parameterized by the $r$-wise variance or quantiles
and the resulting global envelope tests are related to a particular form of the so-called deviation test
\citep{Diggle1979} (Section \ref{sec:devtests}).

Determining the critical bounds and calculating the $p$-values in the new tests are
based on ordering or ranking the observed and simulated functional statistics, which creates a test that
is equivalent to the test \eqref{eq:general envtest}.
In essence, the test based on ordering the statistics is the Barnard's Monte Carlo test \citep{Barnard1963, Dufour2006},
which we extend to the functional statistics case.

Ordering the simulated functions is a general approach that is not restricted
to the orderings that correspond to global envelope tests, although these are the most interesting types for us.
This methodology is closely connected to the methods used in functional data analysis,
where measures of functional depth are typically used to describe the centrality of a curve among a set of curves
\citep[e.g.][]{RamsaySilverman2006, RamsayEtal2009, Lopez-PintadoRomo2009, Lopez-PintadoRomo2011}.
Obviously, these measures from functional data analysis can be readily adapted
to testing spatial hypotheses.
By contrast, the tests proposed or discussed in the present study are strictly
spoken tests of hypotheses on the functions $T(r)$, and thus they are not restricted
to summary functions of point processes, but instead they can be applied to any functional
data, or curves.

The remainder of this paper is organised as follows.
We begin by presenting an example that illustrates the usefulness of the envelope test
and we describe the state-of-the-art in hypothesis testing for spatial point processes (Section \ref{sec:motivating_ex_and_current_tech}).
In Section \ref{sec:MCprinciples}, we describe the theory behind (Barnard's) Monte Carlo testing
in a generalized form, which is applicable to functional statistics.
Next, we introduce two types of global envelope tests that differ in terms of boundary construction
(Sections \ref{sec:rank_test} and \ref{sec:devtests}).
In Section \ref{sec:rankmeasures}, we briefly consider other ways of ordering the functional statistics.
Section \ref{sec:composite} addresses the construction of global envelope tests for composite hypotheses.
The results of extensive simulations to study the power of the proposed tests are presented in Section \ref{sec:simstudy}.
Finally, we illustrate the performance of the proposed methodology with interesting ecological and biological examples (Section \ref{sec:datastudy}),
and we conclude with a discussion (Section \ref{sec:discussion}).

The proposed methods are provided as an R library \emph{spptest},
which can be obtained at \\
\texttt{https://github.com/myllym/spptest}, and they will also be available in the R library spatstat \citep{BaddeleyTurner2005}.

\section{A close look at point patterns testing: a motivational example and current techniques}\label{sec:motivating_ex_and_current_tech}

To illustrate the merits of the new global envelope tests, we revisit
the spatial pattern of intramembranous particles, which was studied previously by \citet{SchladitzEtal2003} (Figure \ref{fig:particles}).
The scientific questions addressed comprise whether the particles ``interact'' or not,
what type of interaction is present, and the interpoint distances where the possible interaction occurs.
\citet{SchladitzEtal2003} tested the complete spatial randomness (CSR) hypothesis, but
we employ a Gibbs hard core model (i.e.,\ the Strauss process where the interaction parameter $\gamma$ equals zero, see Section \ref{sec:ppmodels}) as the null model for this pattern because the particles have a spatial dimension.
We can regard the Gibbs hard core point process as the conditional Poisson process where the condition is that every pair of distinct points is at least some (hard core) distance apart.	
We condition the Gibbs hard core model on the number of points, and thus the only parameter in the model is the hard core distance, which we fix to the minimum distance between two particles, i.e., 5.88 pixels in our sample.

\begin{figure}[htbp]
\centering
\makebox{\includegraphics[width=0.5\textwidth]{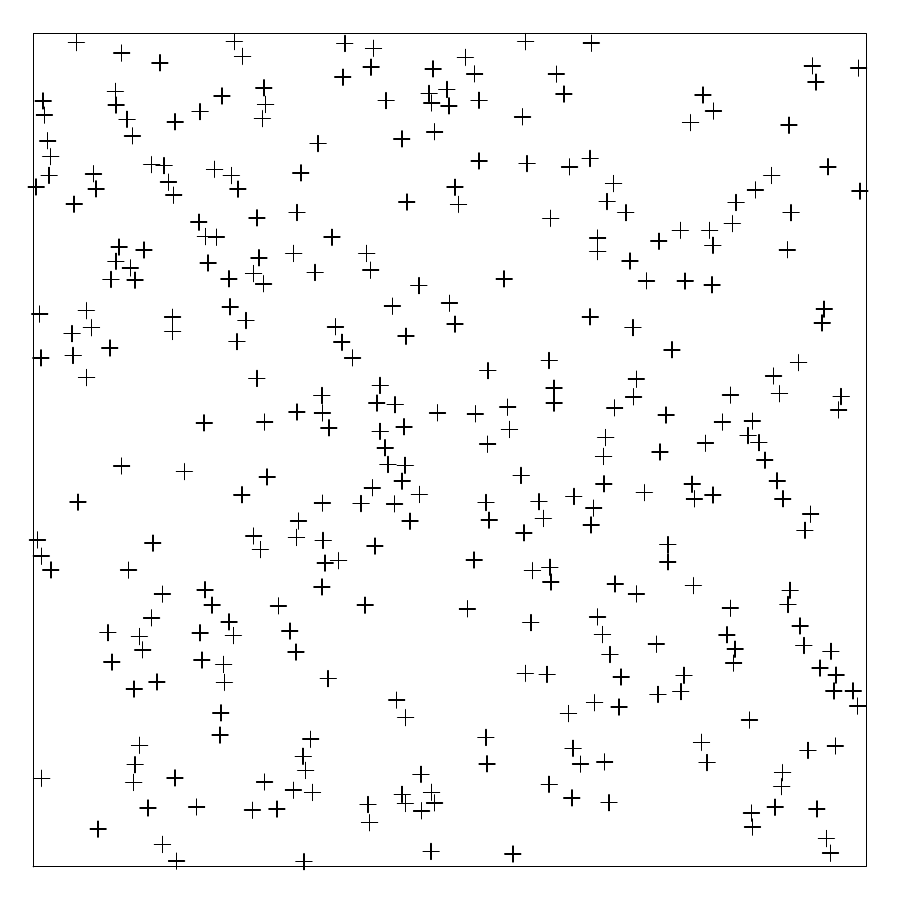}}
\caption{\label{fig:particles}A point pattern of 337 intramembranous particles from an untreated control group observed in an area of 512 $\times$ 512 pixels (336 nm $\times$ 336 nm).}
\end{figure}

To study the interaction between points in our illustrative example, we use
the so-called $L$-function \citep{Ripley1976, Ripley1977, Besag1977}, which
is the most common summary function for point patterns, and we set $I=[6, 125]$
to exclude the smallest interpoint distances within
which no pairs of points can occur, and we take the same upper bound $125$ as \citet{SchladitzEtal2003}.
First, we test the hard core model using conventional methods
and then with the new methods (see Section \ref{sec:datastudy}).

A classical test is the {\em deviation test} introduced by \citet{Diggle1979}.
In this test, the discrepancy between the empirical function $T(r)$
and the corresponding theoretical mean under the null hypothesis is summarized by a
real-valued deviation measure, e.g.,\
the integrated squared difference on $I$.
The same measure $u$ is calculated for all $T_i(r)$ and the measures are ranked.
The null hypothesis is rejected if $u$ for the data
takes an extreme rank, which can be expressed by a $p$-value,
For the pattern in Figure \ref{fig:particles}, we obtained $p=0.03$ for the hard core null model
using the integrated squared difference (with $s=99$).
Thus, the Gibbs hard core model was rejected at the significance level $0.05$ by this test.

In general, a practical shortcoming of the deviation test is that it does not
indicate the distances where the behaviour of the empirical $T(r)$
leads to the rejection of the null hypothesis.
This information is valuable, so it has become popular to compare the empirical function $T_1(r)$ to a pointwise envelope, i.e.,\
the $k$th smallest and largest values of the $s$ simulated functions $T_i(r)$ for each $r\in I$,
as suggested by \citet{Ripley1977} and \citet{BesagDiggle1977}.
Often, this plot is incorrectly taken as a global envelope test in the sense of \eqref{eq:general envtest}.
If the set $I$ is a singleton $\{r_0\}$, the two-sided test rejects the null hypothesis with a probability of $2k/(s+1)$ given that the null hypothesis is true. Thus, the type I error probability is controllable. However, $T$ is usually inspected on an interval $I$, which creates a multiple testing situation. Previously, \citet{Ripley1977} noted that the risk of a type I error has a probability that exceeds $2k/(s+1)$. As discussed in detail by \citet{LoosmoreFord2006} and \citet{GrabarnikEtAl2011},
as well as \citet{BaddeleyEtal2014}, the use of predetermined values for $k$ and $s$ leads to an unpredictable level of error.

A pointwise envelope based on the typical choices $k=1$ and $s=99$ is shown in
Figure \ref{fig:particles-env99} for our example pattern.
This suggests that several distances can be used as evidence against the null model.
However, an approximation of the global type I error probability \citep[see][]{LoosmoreFord2006,GrabarnikEtAl2011} is 0.36,
which is hardly acceptable as a level for a statistical test.

\begin{figure}[htbp]
\centering
\makebox{\includegraphics[width=0.5\textwidth]{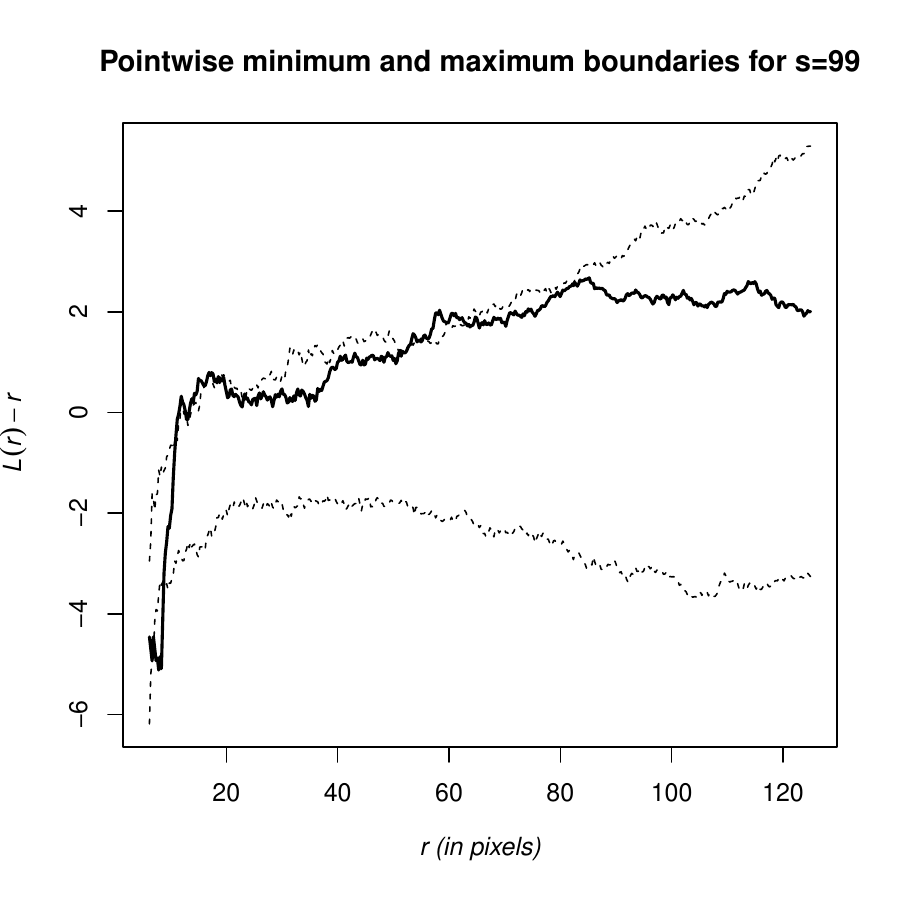}}
\caption{\label{fig:particles-env99}The pointwise minimum and maximum (dashed lines) of the $s=99$ simulated functions $T(r)=\hat{L}(r)$ for the hard core process based on the intramembranous particles data shown in Figure \ref{fig:particles}. The solid line is the empirical function estimated from the data.}
\end{figure}

Some previous studies have considered a controlled global type I error probability for pointwise envelopes, e.g.,
\citet{DavisonHinkley1997} proposed a resampling method for estimating
the global type I error probability of the envelope test for {\em a priori}
selected values of $k$ and $s$.
Independently, \citet{GrabarnikEtAl2011} refined the envelope test in the spirit of
the sequential version of Barnard's Monte Carlo test \citep{BesagClifford1991}
in order to adjust it for simultaneous inference.
Since the global type I error probability of an envelope test for fixed $k$ can
be estimated {\em a posteriori} after a certain number of simulations $s$ have been
completed, the number of simulations $s$ can be adjusted iteratively to obtain
a prescribed global type I error probability $\alpha$.

In the following sections, we propose global envelope tests,
where the significance level $\alpha$ and the number of simulations $s$ can be chosen \emph{a priori}.

\section{Essentials of Barnard's Monte Carlo test}\label{sec:MCprinciples}

In this study, we use a functional test statistic $T(r)$ to test $H_0$.
Typically, the distribution of $T(r)$ is unknown, so Monte Carlo methods have
to be used to determine whether an observed function $T(r)$ is in some sense ``extreme''.
We employ the approach proposed by \citet{Barnard1963} to handle intractable
distributions of test statistics (a similar approach was discussed by \citet{Dwass1957} in the context of a randomized test).

In the following, we consider Monte Carlo goodness-of-fit tests for simple hypotheses. In practical applications, the null hypothesis is often composite, e.g., we want to test whether an observed point pattern is a realization of a parametric model. The testing procedure is then extended by parameter estimation. The composite hypothesis case is addressed in Section \ref{sec:composite}.

The idea behind Barnard's Monte Carlo test is that if $H_0$ is true, the observed $T_1$ and $s$ simulated statistics $T_2,\dots,T_{s+1}$ are identically and independently distributed (in this context, $T$ is some statistic with sample space $S$, but it is not necessarily a function).
Therefore, the null hypothesis can be rejected with the exact probability $\alpha$
if $T_1$ is among the $\alpha(s+1)$ extremal values of $T_i$s, $i=1,\dots,s+1$.
The number of simulations $s$ is selected such that $\alpha (s+1)$ is an integer.

It was noted \citep{BesagClifford1989, Dufour2006} that the independence can be relaxed to exchangeability, which means that the joint probability distribution of the variables $T_1, \ldots, T_{s+1}$ is permutation invariant,
\begin{equation*}
  \Pr\big( (T_1, \dots, T_{s+1})\in A \big)
  = \Pr\big( (T_{\sigma(1)}, \dots, T_{\sigma({s+1})})\in A \big)
\end{equation*}
for any measurable set $A\subseteq S^{s+1}$ and any permutation $\sigma$.
The exchangeability assumption is useful because it allows us to use Barnard's Monte Carlo test, e.g.\
when the test statistic is a deviation measure based on an unknown theoretical expectation \citep[see e.g.][]{Diggle2013, DavisonHinkley1997}.

Provided that the observed and simulated test statistics are independent, or at least exchangeable,
Barnard's Monte Carlo test can be extended to the case where the test statistic $T$ is a functional statistic.
Barnard's Monte Carlo test requires the identification of the $\alpha(s+1)$ most extreme statistics.
This can be done by introducing an ordering on $S$.
In the following, we assume that $T_i\prec T_j$ denotes that $T_i$ is ``more extreme'' than $T_j$.
Furthermore, $T_i \sim T_j$ denotes that $T_i$ is ``as extreme as'' $T_j$, i.e.,\
we assume that there can be ties and that the ordering is weak in this general case.

First, we assume that the functions $T_i$ can be strictly ordered, i.e.,\ all pairs of functions are comparable and for each pair, either $T_i \prec T_j$ or $T_j \prec T_i$ for $i\neq j$. The strict ordering of the functional statistics allows us to build an {\it exact} test based on the following Lemma.

\begin{lemma}\label{lemma:funcstat_test}
Let $T_1, \dots, T_{s+1}$ be exchangeable functional statistics and let $\prec$ be an ordering with
\begin{equation}\label{eq:unique_ordering_condition}
Pr( T_i \prec T_j  \; \text{or}\; T_j  \prec T_i)= 1   \quad \text{for all pairs} \; i \neq j,
\end{equation}
and let
\begin{equation}\label{eq:p-value general}
 p = \frac{1}{s+1}\left(1 + \sum_{i=2}^{s+1} \1(T_i \prec T_1) \right).
\end{equation}
Then, the test
\begin{equation}\label{eq:MCsigtest}
  \varphi(T_1) = \1\big( p \leq \alpha \big) = \1\bigg( 1 + \sum_{i=2}^{s+1}\1( T_i\prec T_1)\  \leq \alpha(s+1) \bigg).
\end{equation}
rejects the null hypothesis at the prescribed significance level, $\alpha$, provided that $\alpha (s+1)$ is an integer.
\end{lemma}

\begin{proof}
Let $A_j= 1+\sum_{i=1, i \neq j}^{s+1} \1(T_i \prec T_j)$ be the rank of the functional statistic $T_j$ among the $T_i$s.
The unique ordering condition \eqref{eq:unique_ordering_condition} implies that $(A_1,\dots,A_{s+1})$ is a permutation of $(1,\dots, s+1)$ almost surely. Then, for $j=1, \dots, s+1$,
$$
 Pr( A_j = k) = 1/(s+1) \quad \text{for}\; k=1, \ldots, s+1,
$$
because the ranks $A_1, \dots, A_{s+1}$ are clearly exchangeable. Immediately, this implies that
\begin{equation}\label{eq:notiesexact}
  \E \varphi(T_1) =\E\big( \1\big( p \leq \alpha \big)\big) = Pr(A_1 \leq \alpha(s+1)) = \frac{\alpha(s+1)}{s+1} = \alpha.
\end{equation}
\end{proof}

The equation \eqref{eq:p-value general} defines the (Barnard) $p$-value.

It is clear that the choice of ordering depends on the features of a functional statistic $T$ that are of interest.
Some of the orderings discussed in this study are weak,
which means that for some pairs of functions, we cannot decide which
is more extreme.
If ties do not have zero probability in the sample $(T_1, \dots, T_{s+1})\in S^{s+1}$,
then we can break them to obtain an exact test as follows.
Let us introduce the associate variables $M_i, i=1, \dots, s+1$, which can be strictly and totally ordered. Then, a strict ordering $\prec_{\text{Lex}}$ of the pairs $(T_i, M_i)$, $i=1, \ldots, s+1$ exists by means of the lexicographic ordering:
$$
(T_i, M_i) \prec_{\text{Lex}} (T_j, M_j) \Leftrightarrow \{ T_i \prec T_j \,\,\text{or}\,\, (T_i \sim T_j \,\,\text{and}\,\, M_i < M_j)\}.
$$
Now, Lemma \ref{lemma:funcstat_test} can be applied to the pairs $(T_1, M_1), \dots, (T_{s+1}, M_{s+1})$, which are all comparable and strictly ordered. In a particular case, $M_i$ can be continuous random variables that are independent of $T_i$, as considered by \citet[Prop.\ 2.4]{Dufour2006}.

\section{The rank envelope test}\label{sec:rank_test}

\subsection{The ``extreme rank'' depth measure}\label{sec:R_i}

In Sections \ref{sec:intro} and \ref{sec:motivating_ex_and_current_tech}, we discussed the conventional pointwise envelopes,  i.e., envelopes defined by the bounding curves
\begin{equation}\label{kth_envelopes}
T^{(k)}_{\low}(r)= \underset{i=1,\dots,s+1}{{\min}^{k}} T_i(r)
\quad\text{and}\quad
T^{(k)}_{\upp}(r)= \underset{i=1,\dots,s+1}{{\max}^{k}} T_i(r) \quad \text{for } r\in I,
\end{equation}
where $\min^k$ and $\max^k$ denote the $k$-th smallest and
largest values, respectively, and $k=1,2,\dots,\lfloor (s+1)/2\rfloor$.
In general, for a fixed number $s$ of simulations and a fixed $k$, the type I error probability of a global envelope test using these $k$-th lower and upper rank curves is unknown a priori. However, as explained in this section, it is possible to calculate a $p$-value based on the position of the observed curve $T_1$ within the $k$-th envelopes
and to construct a global envelope test with (approximately) controlled error. To this end, the curves $T_i$ are ordered
(we temporarily assume that there are no \emph{pointwise ties}, i.e.,\ ties in the values $T_i(r)$ for $r\in I$)
according to the largest $k$ for which they are still present in the $k$-th envelope, i.e.,
\begin{equation}\label{rank_measure_expl1}
R_i := \max \left\{ k :\  T^{(k)}_{\low}(r) \leq T_i(r) \leq T^{(k)}_{\upp}(r) \ \  \text{for all } r\in I\right\}.
\end{equation}
The value $R_i$, which we call the {\it extreme rank}, is a depth measure that represents
the apparent ``extremeness'' of the curve $T_i$ in the bundle of functions $T_1,\dots, T_{s+1}$.

Formally, we define the \emph{extreme rank} depth measure as follows.
\begin{enumerate}
\item For each $r$, let $R_i^{\uparrow}(r)$ and $R_i^{\downarrow}(r)$, $i=1,\dots,s+1$, denote the ranks of the values
$T_i(r)$, $i=1,\dots,s+1$,
from the smallest value (with rank 1) to the largest (rank $s+1$),
and from the largest (1) to the smallest ($s+1$), respectively.
(In the case of pointwise ties in the values $T_i(r)$, we use the mid-rank
in order to weigh down the influence of distances where many $T_i(r)$ coincide in the test.
An alternative is to use the maximum rank.)
\item Let
\begin{equation}\label{extreme_rank_r}
R_i^*(r) = \min(R_i^{\uparrow}(r), R_i^{\downarrow}(r))
\end{equation}
denote the $r$-wise rank of $T_i(r)$ adjusted to a two-sided variant of the test.
When $R_i^*(r)$ is smaller, $T_i(r)$ is more extreme among the other values for the argument $r$.
\item Then, we define the \emph{extreme rank} as
\begin{equation}\label{rank_measure}
 R_i = \min_{r\in I} R_i^*(r).
\end{equation}
\end{enumerate}
The term \emph{extreme rank} for the depth measure $R_i$ is based on the fact that $R_i$ is the most extreme of the $r$-wise ranks of $T_i(r)$.

In practice, the functions $T_i$ are not evaluated on a continuous interval $I$, but instead a discretized version $I_\fin$ is employed, which makes it feasible to calculate the $R_i$s.

\subsection{Calculation of the $p$-values}\label{sec:rank_p}

The extreme ranks $R_i$ introduce an ordering of the functions
$T_i(r)$ in ascending order from the most extreme to the most typical under the null model.
Thus, it can be used for Monte Carlo testing, as described in Section \ref{sec:MCprinciples}.
However, the set of all $R_i$ contains ties because $\max R_i < s+1$, and
thus the functions $T(r)$ can be only weakly ordered.

A method for building an exact test with tied functions $T_i(r)$ is described in Section \ref{sec:MCprinciples}
and specifically for the rank envelope test in Section \ref{sec:goldfirst}.
Exactness
is desirable from a theoretical viewpoint and to perform power comparisons, but
a practical approach to handling ties is to report a range of $p$-values, similar
to the methods proposed for Barnard's test with discrete test statistics \citep[e.g.,][]{BesagClifford1989, BesagClifford1991}.

The range of $p$-values encompasses the most liberal and conservative $p$-values,
which in the case of the rank envelope test are calculated as
\begin{equation}\label{eq:pvalue_interval}
p_- = \frac{1}{s+1} \sum_{i=1}^{s+1} \1( R_i < R_1 )
\quad \text{and} \quad
p_+ = \frac{1}{s+1} \sum_{i=1}^{s+1} \1( R_i \leq R_1 ),
\end{equation}
respectively. These values provide the lower and upper bounds for the $p$-value of the test.
\begin{proposition}
When $s$ is chosen such that $\alpha(s+1)$ is an integer, the test corresponding to the lower $p$-value,
\begin{equation}\label{eq:test_lib}
\varphi_\lib(T_1)=\1(p_- \leq \alpha)
\end{equation}
is liberal, i.e.,\ $\E\varphi_\lib(T_1)>\alpha$, whereas the test
\begin{equation}\label{eq:test_cons}
\varphi_\cons(T_1)=\1(p_+ \leq \alpha)
\end{equation}
is conservative, i.e.,\ $\E\varphi_\cons(T_1)\leq \alpha$.
\end{proposition}
\begin{proof}
  It is sufficient to condition on the set $\{R_1,\dots,R_{s+1}\}$ of extreme ranks and to show that
  \begin{equation}\label{eq:prooflibcons:condmean}
  \E(\varphi_\lib(T_1) \mid \{R_1,\dots,R_{s+1}\}) > \alpha \quad \text{and}\quad
   \E(\varphi_\cons(T_1) \mid \{R_1,\dots,R_{s+1}\}) \leq \alpha.
  \end{equation}
   As usual, let $R_{(1)}\leq R_{(2)}\leq\dots\leq R_{(s+1)}$ denote the ordered set and write
  \[
   A(x) = \sum_{i=1}^{s+1} \1(R_i < x) \quad \text{and}\quad    B(x) = \sum_{i=1}^{s+1} \1(R_i \leq x).
  \]
  Then, $A(R_{(k)})<k$ and $B(R_{(k)})\geq k$, and thus
   \begin{equation}\label{eq:prooflibcons:indicator}
   \sum_{k=1}^{s+1}\1({A(R_{(k)})}\leq x) > \sum_{k=1}^{s+1}\1(k\leq x) \quad \text{and}\quad
   \sum_{k=1}^{s+1}\1({B(R_{(k)})} \leq x) \leq \sum_{k=1}^{s+1}\1(k\leq x)
  \end{equation}
    Under $H_0$, the observed $R_1$ takes any place in the ordered set with equal probability; therefore,
  \begin{align*}
    \E(\varphi_\lib(T_1) \mid \{R_1,\dots,R_{s+1}\}) &=  \Pr({A(R_1)}/{(s+1)}\leq\alpha \mid \{R_1,\dots,R_{s+1}\})
    \\&=  \frac{1}{s+1}\sum_{k=1}^{s+1}\1({A(R_{(k)})}\leq\alpha(s+1))
    \\&>  \frac{1}{s+1}\sum_{k=1}^{s+1}\1(k\leq\alpha(s+1))=\alpha.
  \end{align*}
The result follows from averaging over the set $\{R_1, \dots, R_{s+1}\}$.
In the same manner, we can use \eqref{eq:prooflibcons:indicator} to show that $\varphi_\cons$ is conservative.
\end{proof}

A $p$-interval can be safely interpreted as a test at the level $\alpha$
by the following rule: if $p_+\leq \alpha$, the null
hypothesis is clearly rejected, and if
$p_->\alpha$, there is no evidence for rejecting
$H_0$.
When $p_-\leq \alpha <p_+$,
it is not clear whether to reject or not.
From a theoretical viewpoint, a unique $p$-value is desirable.
A standard approach is to choose $p$ randomly in the $p$-interval.
In order to avoid a random decision, we describe a finer ordering in Section \ref{sec:goldfirst},
which is strict in virtually all practical cases.

\subsection{$100\cdot(1-\alpha)$\% global rank envelope}\label{sec:rank_envelope}

The $p$-interval given by the extreme ranks $R_i$
already determines the test result of the rank test (if $\alpha\notin[p_-, p_+)$).
However, the graphical interpretation given by an envelope is essential for determining the
possible reasons for rejecting the null hypothesis.

The significance tests $\varphi_\lib$ and $\varphi_\cons$ given by \eqref{eq:test_lib} and \eqref{eq:test_cons} map directly onto the events that the observed
curve $T_1$ exceeds or at least touches the $k_\alpha$-th rank envelope, where the critical rank $k_\alpha$ is determined as follows:
\begin{equation}\label{eq:kalpha}
  k_\alpha = \max\left\{k:\  \sum_{i=1}^{s+1}\1(R_i<k) \leq \alpha(s+1)\right\}.
\end{equation}
With this critical rank, the corresponding liberal and conservative global envelope tests are
\begin{equation*}
\varphi_{\env,\lib}(T_1)=\1\big(\,\exists \; r\in I: T_1(r)\notin (T^{(k_{\alpha})}_{\low}(r), T^{(k_{\alpha})}_{\upp}(r))\, \big)
\end{equation*}
and
\begin{equation}\label{eq:rank-envelope-test-cons}
\varphi_{\env,\cons}(T_1)=\1\big(\,\exists \; r\in I: T_1(r)\notin [T^{(k_{\alpha})}_{\low}(r), T^{(k_{\alpha})}_{\upp}(r)]\,\big).
\end{equation}

The following theorem states that inference based on the $p$-interval and the global envelope specified by $T^{(k_{\alpha})}_{\low}(r)$ and $T^{(k_{\alpha})}_{\upp}(r)$ are equivalent. Therefore, we can loosely refer to the $k_\alpha$-th envelope as the $100\cdot(1-\alpha)$\% global \emph{rank envelope}.
\begin{theorem}\label{thm:envelope-vs-pinterval}
Let $p_-$ and $p_+$ be as given in \eqref{eq:pvalue_interval}, and $T^{(k_{\alpha})}_{\low}(r)$ and $T^{(k_{\alpha})}_{\upp}(r)$ define the $100\cdot(1-\alpha)$\% global rank envelope. Then, assuming that there are no pointwise ties with probability $1$, it holds that:
\begin{enumerate}[(a)]
 \item $T_1(r) < T^{(k_{\alpha})}_{\low}(r)$ or $T_1(r) > T^{(k_{\alpha})}_{\upp}(r)$ for some $r\in I$
iff $p_+ \leq \alpha$, in which case the null hypothesis is rejected;
 \item $T^{(k_{\alpha})}_{\low}(r) < T_1(r) < T^{(k_{\alpha})}_{\upp}(r)$ for all $r\in I$ iff $p_- > \alpha$, and thus the null hypothesis is not rejected;
 \item and the data function $T_1(r)$ does not exit the envelope, but it coincides with
$T^{(k_{\alpha})}_{\low}(r)$ or $T^{(k_{\alpha})}_{\upp}(r)$
for some $r\in I$ iff $p_- \leq \alpha < p_+$.
\end{enumerate}
(In the case of pointwise ties, $T_1$ may coincide with $T^{(k_{\alpha})}_{\low}$ or $T^{(k_{\alpha})}_{\upp}$ at the distances $r$ with ties, even though $p_- > \alpha$.)
\end{theorem}
\begin{proof}
Let
\begin{equation*}
a(k) = \frac{1}{s+1} \sum_{i=1}^{s+1} \1(R_i < k) \quad\text{for }\ \  k=1,2,\dots
\end{equation*}
denote the proportion of functions $T_i(r)$
that are either below $T^{k}_{\low}(r)$ or exceed $T^{k}_{\upp}(r)$ at some value $r\in I$. Then, by \eqref{eq:pvalue_interval}, $p_+=a(R_1+1)$ and $p_-=a(R_1)$. From \eqref{eq:kalpha}, we find that $a(k_\alpha)\leq\alpha$ and $a(k_\alpha+1)>\alpha$ (for all $k_\alpha \in\{R_1,\dots,R_{s+1}\}$). Thus:
(a) $p_+ \leq \alpha $ iff $R_1 < k_\alpha$, which means that $T_1(r)$ is strictly outside the global rank envelope.
(b) $p_- > \alpha $ iff $R_1 \geq k_\alpha+1$, which means
that $T^{(k_{\alpha})}_{\low}(r) < T_1(r) < T^{(k_{\alpha})}_{\upp}(r)$ for all $r\in I$.
(c) $R_1 = k_\alpha$ iff $p_- = a({k_\alpha}) \leq \alpha$ and $p_+ = a({k_\alpha+1}) > \alpha$, which means that $T_1(r)$ coincides with the global rank envelope for some $r\in I$.
 \end{proof}

Note that for the rank envelope test there is no need to know or estimate
the expectation of the function $T(r)$ under $H_0$. However, for visualization purposes,
it often makes sense to plot the expectation $T_0(r)$ together with $T_1(r)$ and
$(T^{(k_{\alpha})}_{\low}(r), T^{(k_{\alpha})}_{\upp}(r))$.

{\em Remark 1.} For the test functions used in this study, pointwise ties usually appear only when an inappropriate interval $I$ is selected, e.g.,\ if the null model is the Gibbs hard core process and the lower bound of the interval $I$ selected is smaller than the hard core distance.

{\em Remark 2.}
The data function $T_1$ is included in the set of functions
used to define the $k$-th lower and upper curves \eqref{kth_envelopes},
in which case the global envelope corresponds to the $p$-values, as explained above.
In the conventional and refined envelope tests, the envelopes are instead
constructed from simulations only \citep[see][]{Ripley1977, GrabarnikEtAl2011}.
The effect of the inclusion (or exclusion) of $T_1$ from $T_i$
when calculating the global envelope is obviously negligible
when the number of simulations $s$ is large.

\subsection{Number of simulations}

To make the ambiguous result $p_-\leq\alpha<p_+$ unlikely, the expected width of the $p$-interval should be small.
This width depends on the number of simulations $s$, the smoothness properties of $T(r)$, and the value of $R_1$. Based on our investigations presented in Appendix \ref{sec:s}, our general recommendation is to use at least $s=2500$ simulations for testing at the significance level $\alpha=0.05$.
According to the experiments presented in the Appendix \ref{sec:s}, the number of simulations $s \geq 2500$ tends to lead to widths of the $p$-interval that are smaller than 0.01 close to the nominal level $\alpha=0.05$.
We explored this with commonly used summary functions for point processes.

\section{Global scaled maximum absolute difference (MAD) envelope tests}\label{sec:devtests}

The rank envelope test with a reasonably narrow $p$-interval is not feasible if it is not possible to run sufficiently many simulations. In this case, we can employ {\em scaled MAD} envelope tests,
where the critical bounds are of the form
\begin{equation}\label{u_envelopes_st}
T^u_\low(r) = T_0(r) - u\cdot \sqrt{\var(T(r))} \quad\text{and}\quad T^u_\upp(r) = T_0(r) + u\cdot \sqrt{\var(T(r))}
\end{equation}
for the {\em studentized} MAD envelope test and
\begin{equation}\label{u_envelopes_qdir}
T^u_\low(r) = T_0(r) - u\cdot |\underline T(r) - T_0(r)| \quad\text{and}\quad T^u_\upp(r) = T_0(r) + u\cdot |\overline T(r) - T_0(r)|
\end{equation}
for the {\em directional quantile} MAD envelope test.
$T_0(r)$ and $\var(T(r))$ denote the expectation and variance of $T(r)$ under the null model,
and $\overline T(r)$ and $\underline T(r)$ are the $r$-wise 2.5$\%$ upper and
lower quantiles of the distribution of $T(r)$ under $H_0$, respectively.
The expectation $T_0(r)$ and the quantities $\var(T(r))$, $\overline T(r)$ and $\underline T(r)$
can be estimated from simulations of the null model \citep[see e.g.][]{Diggle2013}
if they are not known analytically.

The above envelopes are parameterized with respect to $u$ and the aim is to find $u$ such that
the global envelope test based on these envelopes on $I$ has the desired significance level $\alpha$.

The shape of the scaled MAD envelopes is determined only by the variance or quantiles,
but an advantage is that they do not require a large number of simulations.
These tests have corresponding deviation tests, for which
the number of simulations recommended in previous studies
\citep[see e.g.][]{Hope1968, Marriott1979, Diggle2013} is rather small, i.e.,
$s=99$ and $s=199$ are popular choices.

Next, we recall the deviation test and we then show the correspondence of the scaled MAD envelope tests
to scaled MAD deviation tests.

\subsection{Deviation tests}

Now, we recall how deviation tests work. A deviation measure,
e.g.,\ the maximum absolute difference (MAD) measure
\begin{equation}\label{Dinfty_classical}
u = \max_{r \in I} \Big| T(r) - T_{0}(r) \Big|
\end{equation}
or the integral measure
\begin{equation}\label{DL2_classical}
u = \int_I \left(T(r) - T_{0}(r)\right)^2 \text{d}r,
\end{equation}
is selected and computed for the observed pattern ($u_1$) and
for each simulated pattern ($u_2,\dots,u_{s+1}$).
The deviation measures $u_i$, $i=1,\dots,s+1$, define a strict ordering of the functions $T_i(r)$,
where
\begin{equation*}
T_i\prec T_{j} \Longleftrightarrow u_i > u_{j}.
\end{equation*}
Since the test statistic of the deviation test is continuous and ties have zero probability,
the $p$-value of the deviation test is obtained by \eqref{eq:p-value general}
and the corresponding significance test
has the level $\alpha$ for simple hypotheses for any number of simulations $s$.

A general disadvantage of the deviation test is the lack of graphical interpretation.
However, as we discuss in the following,
it is possible to construct a global envelope for the MAD test.

The global envelope based on the classical MAD measure \eqref{Dinfty_classical}
was introduced by \citet{Ripley1981}.
This envelope has a constant width over the distances $r\in I$.
Therefore, it is not sufficiently flexible to represent the behaviour of functions $T_i(r)$ for different distances
$r \in I$ unless the distribution of $T(r)$ is symmetric and the same for all $r$.

Instead, we consider two scaled MAD measures that take
the unequal variance of $T(r)$ and the asymmetric distribution of $T(r)$ into account, i.e.,
\begin{equation}\label{Dinfty_st}
u = \max_{r \in I} \Big| \frac{ T(r) - T_{0}(r) }{ \sqrt{\var(T(r))} } \Big| ,
\end{equation}
and
\begin{equation}\label{Dinfty_qdir}
u = \max_{r \in I} \left( \1(T(r) \geq T_{0}(r)) \frac{ T(r) - T_{0}(r) }{ |\overline T(r) - T_0(r)| } + \1(T(r) < T_{0}(r)) \frac{ T(r) - T_{0}(r) }{ |\underline T(r)-T_0(r)| } \right),
\end{equation}
respectively.
The corresponding scalings can be used for $T(r) - T_{0}(r)$ in the integral measure \eqref{DL2_classical},	 leading to the studentized and directional quantile integral deviation tests, respectively.
Deviation measures similar to \eqref{Dinfty_st} and \eqref{Dinfty_qdir} have been used previously, e.g.,\ by \citet{BaddeleyEtal2000b}, \citet{MollerBerthelsen2012} and \citet{MyllymakiEtal2015}.
In particular, \citet{MyllymakiEtal2015} proposed using the above scalings in the case of symmetrically and asymmetrically distributed $T(r)$, respectively, in order to make the contributions of the residuals $T(r)-T_0(r)$ for different distances $r\in I$ more equal in the test and to improve the performance compared with the classical measures \eqref{Dinfty_classical} and \eqref{DL2_classical}.
In the following, we propose envelopes based on the {\em scaled} MAD measures.
Note that the scalings in \eqref{Dinfty_st} and \eqref{Dinfty_qdir} are general, whereas other scalings have been proposed for specific functions, e.g.,\ \citet{HoChiu2006} advocated scaling for the $K$-function using adapted distance-dependent intensity estimators.

\subsection{$100\cdot(1-\alpha)$\% global envelope based on the scaled MAD}

The following theorem states that a global envelope test exists that corresponds to the
scaled MAD deviation tests based on \eqref{Dinfty_st} and \eqref{Dinfty_qdir}.
\begin{theorem}
Let $u_{\alpha}$ be the $\alpha(s+1)$-th largest value of the $u_i$s
defined by \eqref{Dinfty_st} or \eqref{Dinfty_qdir}.
For the global envelope test \eqref{eq:general envtest}
where the bounds $T_\low(r)$ and $T_\upp(r)$ are specified by the curves
\eqref{u_envelopes_st} or \eqref{u_envelopes_qdir} with $u = u_{\alpha}$ plugged in,
it holds that
\begin{equation*}
\E \varphi_{\env}(T_1) = \E \varphi_{\dev}(T_1),
\end{equation*}
where
\begin{equation*}
 \varphi_{\dev}(T_1) = \1\bigg( \sum_{i=1}^{s+1}\1(u_i\geq u_1)\  \leq \alpha(s+1) \bigg)
\end{equation*}
for the deviation measures \eqref{Dinfty_st} or \eqref{Dinfty_qdir}, respectively.
\end{theorem}

\begin{proof}
Note that a deviation test rejects the null hypothesis if $u_1\geq u_{\alpha}$.
For the measure \eqref{Dinfty_st}, this yields
\begin{equation*}
T_1(r) \leq T_0(r) - u_{\alpha} \sqrt{\text{var}(T(r))} \quad \text{or} \quad
T_1(r) \geq T_0(r) + u_{\alpha} \sqrt{\text{var}(T(r))}
\end{equation*}
and, for \eqref{Dinfty_qdir},
\begin{equation*}
T_1(r) \leq T_0(r) - u_{\alpha} |\underline T(r) - T_0(r)| \quad \text{or} \quad
T_1(r) \geq T_0(r) + u_{\alpha} |\overline T(r) - T_0(r)|.
\end{equation*}
Thus, there is one-to-one correspondence.
\end{proof}

Finding $u$ in \eqref{u_envelopes_st} or \eqref{u_envelopes_qdir} such that the global envelope test has an appropriate level is simple: we just need to take the $\alpha(s+1)$-th largest value among the $u_i$s. 

For the classical case \eqref{Dinfty_classical} without any scaling,
the $100(1-\alpha)$\% global envelope is simply $(T_0(r) - u_{\alpha}, T_0(r) + u_{\alpha})$,
where $u_{\alpha}$ is the $\alpha(s+1)$-th largest value of the $u_i$s
defined by \eqref{Dinfty_classical}.
However, the widths of the scaled MAD envelopes vary with $r$, where
the shape of the envelope around $T_0(r)$ is determined either by the standard
deviation or the quantiles $\underline T(r)$ and $\overline T(r)$.
The studentized envelope is symmetric around $T_0(r)$, whereas
the directional quantile envelope can be asymmetric.

\section{Other functional depth orderings}\label{sec:rankmeasures}

It is obvious that other orderings can be constructed from the $r$-wise ranks
of the functions $T_i(r)$ in addition to that based on the extreme rank.
Other orderings are unlikely to have a global envelope representation, but
they may have higher power than the measure \eqref{rank_measure} in some situations.

In Section \ref{sec:MHRD_MBD}, we consider two functional depth measures from previous studies
and we also include them in our simulation study in Section \ref{sec:simstudy}.
More importantly, in Section \ref{sec:goldfirst}, we introduce a new ordering, which is closely connected to the extreme rank ordering. We consider this {\em rank length} or {\em rank count} ordering in detail and show its connection with the rank envelope test.

\subsection{Rank length and rank count orderings}\label{sec:goldfirst}

Instead of taking the most extreme pointwise rank ($R_i$) that a curve $T_i$ attains in the interval $I$ as a measure of the centrality of the curve within the sample, we could use a more comprehensive summary of the pointwise ranks. Thus, we propose to consider ``how often'' the ranks $R^*_i(r)=k$, $k=1,\dots, \lfloor(s+2)/2\rfloor$ have been reached by the curve $T_i$. In a continuous setting, this can be expressed by a vector $\Lvec_i$ of \emph{rank lengths} $L_{ik}$,
\begin{equation}\label{eq:rank lengths}
\Lvec_i=(L_{i1},\dots, L_{i\, \lfloor(s+2)/2\rfloor}),\qquad L_{ik}=\int_I \1(R_i^*(r)=k) \dd r.
\end{equation}
The reverse lexical ordering of the rank length vectors $\Lvec_i$,
\begin{equation}\label{eq:rank lengths lexical}
\Lvec_i \prec\Lvec_j \quad \Longleftrightarrow\quad
  \exists \; n\leq \lfloor(s+2)/2\rfloor: L_{i k} = L_{j k}\ \forall \ k < n,\  L_{i n} > L_{j n}
\end{equation}
induces a natural ordering of the curves $T_i$, where
a curve that scores a larger number of small pointwise ranks is more extreme.
In practice, it is very difficult to calculate the integrals $L_{ik}$, but instead we can evaluate the curves at discrete points $r\in I_\fin$ and count the ranks. Thus, we obtain a vector $\Nvec$ of \emph{rank counts},
\begin{equation}\label{eq:rank counts}
\Nvec_i=(N_{i1},\dots, N_{i\, \lfloor(s+2)/2\rfloor}),\qquad N_{ik}=\sum_{r \in I_\fin} \1(R_i^*(r)=k).
\end{equation}

Obviously, both the extreme rank lengths and counts refine the ordering given by the most extreme ranks, $R_i = \min_{r\in I} R_i^*(r)$ or
$R_i = \min_{r\in I_\fin} R_i^*(r)$, and it holds that
\begin{equation}\label{eq: R smaller implies LN smaller}
  R_i < R_j \implies \Lvec_i \prec \Lvec_j\  \text{and}\   \Nvec_i \prec \Nvec_j.\
\end{equation}
The converse is not true.

Since the rank lengths $L_{ik}$ are continuous, they define a strict ordering of the functions $T_i$.
Thus, the $p$-value, $p_\Lgf$, is obtained in an analogous manner to \eqref{eq:p-value general},
and the corresponding significance test at the nominal level $\alpha$,
\begin{equation*}
\varphi_\Lgf(T_1)=\1(p_\Lgf \leq\alpha), 
\end{equation*}
has the correct level, according to Lemma \ref{lemma:funcstat_test}.
Ties may occur for the discretized version of the rank lengths, i.e.,\ for the rank counts $\Nvec_i$.
However, the ties in $\Nvec_i$ are also very unlikely if the set $I_\fin$ is not very small.
If ties occur, they can be broken by randomization.
Thus, in practice, the Monte Carlo $p$-value $p_\Ngf$ is also obtained in an analogous manner to \eqref{eq:p-value general}.
Furthermore,
\begin{equation*}
\E \left(\lim_{|I_\fin|\rightarrow\infty} \varphi_\Ngf(T_1)\right)=\E\left(\lim_{|I_\fin|\rightarrow\infty} \1(p_\Ngf \leq\alpha)\right) = \E \varphi_\Lgf(T_1) = \alpha.
\end{equation*}

The following trivial Proposition \ref{prop:goldfirst breaks ties} suggests that the $p$-values obtained by rank count ordering may be used as a refinement of the $p$-interval of the rank envelope test based on the extreme ranks $R_i$.

\begin{proposition}\label{prop:goldfirst breaks ties}
  Let $p_-$ and $p_+$ denote the $p$-values defined for extreme rank ordering in \eqref{eq:pvalue_interval}. Then, it is always the case that
  \begin{equation*}
    p_- < p_N \leq p_+.
  \end{equation*}
\end{proposition}

Furthermore, the following theorem states that asymptotically, the tests based on the rank counts and the rank envelope test are equivalent.
\begin{theorem}\label{thm:gf-envelope}

For the tests $\varphi_\Ngf(T_1)=\1(p_\Ngf \leq\alpha)$ and \eqref{eq:rank-envelope-test-cons} based on the discretized $r$-interval $I_\fin$,
\begin{equation*}
\varphi_{\env,\cons}(T_1)=\1\big(\,\exists \; r\in I_\fin: T_1(r)\notin [T^{(k_{\alpha})}_{\low}(r), T^{(k_{\alpha})}_{\upp}(r)]\,\big),
\end{equation*}
it holds that
  \begin{equation}\label{eq:env smaller gf}
  \E\varphi_{\env,\cons}(T_1) \leq \E\varphi_N(T_1)
   \end{equation}
  \text{ and in the limit as the number of simulations $s\to\infty$}
  \begin{equation}\label{eq:env goesto gf}
    \E\varphi_{\env,\cons}(T_1) \to \E\varphi_N(T_1) .
    \end{equation}
\end{theorem}
\begin{proof}
 See Appendix \ref{sec:proof_thm6.2}.
\end{proof}

\subsection{Modified half-region and band depths}\label{sec:MHRD_MBD}

The modified band depth (MBD) and the modified half-region depth (MHRD)
were introduced by \citet{Lopez-PintadoRomo2009, Lopez-PintadoRomo2011},
who used these depths to find the deepest (most central) curves among a set of curves.
The MHRD of the curve $T_i(r)$ among
all the curves $T_i(r)$, $i=1,\dots, s+1$, is proportional to
\begin{equation}\label{mhrd_measure}
MHRD_i = \min\left(\int_{r \in I} R_i^{\uparrow}(r) \, \mathrm{d}r, \, \int_{r \in I} R_i^{\downarrow}(r) \, \mathrm{d}r\right),
\end{equation}
where $R_i^{\uparrow}(r)$ and $R_i^{\downarrow}(r)$ are the ranks of $T_i(r)$ as in Section \ref{sec:R_i} (and the maximum treatment of the pointwise ties is taken).
\citet{Lopez-PintadoRomo2011} introduced this depth \eqref{mhrd_measure} as a faster alternative
to MBD, which has a similar integral nature to MHRD.
In a special case which was used in our study, the MBD measure can be interpreted as
the average proportion of $r \in I$ for which $T_i(r)$ lies between
a pair of $T_j(r), j=1,\dots,s+1$, averaged over all distinct pairs.

The above depths introduce orderings for the functions $T_i(r)$:
\begin{equation*}
T_i\prec T_{i'} \Longleftrightarrow MBD_i < MBD_{i'} \quad\text{or}\quad MHRD_i < MHRD_{i'}.
\end{equation*}
Thus, they can be used for testing, where
the $p$-value can be calculated according to \eqref{eq:p-value general}
since ties are unlikely to occur.
However, to the best of our knowledge,
it is not possible to construct a global envelope from the curves $T_i(r)$ such that the envelope test
would correspond to the test based on the MHRD or MBD, although
different visualization is possible by utilising the methods in \citet{SunGenton2011}.

\section{Global envelopes for composite hypotheses}\label{sec:composite}

Often, we are not interested in testing the simple hypothesis that an observed point pattern is a realization of a process with a completely specified distribution, but instead we are concerned with the composite hypothesis that it derives from a parametric model with unknown parameters. Usually, we would fit the model to the observation and use the fitted model in the Monte Carlo test. This ``plug-in'' approach can lead to conservative tests, as discussed based on simulation experiments in Section \ref{sec:simstudy_composite_hypothesis} and in \citet{BaddeleyEtal2014}.
In practice, the conservativeness of these tests is often accepted, but a problem is the decreased power of the test in this case. This problem of conservativeness is still regarded as an unresolved research problem \citep[e.g.][]{BayarriBerger2000, RobinsEtal2000}, although several strategies have been designed to address this problem \citep[e.g.][]{RobinsEtal2000, BrooksEtal1997}. We adapt the method of \citet{DaoGenton2014} to our global envelope tests.

\citet{DaoGenton2014} proposed a method for correcting the level of the classical deviation test when testing composite hypotheses. Their method involves estimating the distribution of Monte Carlo $p$-values based on further simulations  given the parameter estimates and using the empirical $\alpha$-quantile of this distribution, $\alpha^*$, to construct the significance test.
According to \citet{DaoGenton2014}, the test
\begin{equation}\label{eq:DaoGentontest}
  \varphi^*(T_1)=\1(p \leq \alpha^*)
\end{equation}
has the correct level $\alpha$.
Their method is not restricted to tests based on deviation measures, but instead it could be used with any statistics that lead to a unique ordering of point patterns.

In particular, the method can be applied to the scaled MAD measures
\eqref{Dinfty_st} and \eqref{Dinfty_qdir}, and thus to the scaled MAD envelope tests.
Thus, after finding $\alpha^*$, we need to find $u_{\alpha^*}$, which is
the $\alpha^*(s+1)$-th largest value among the $u_i$s.
Then, the corresponding {\em adjusted} studentized and directional quantile envelopes can be obtained
by plugging $u_{\alpha^*}$ in \eqref{u_envelopes_st} and \eqref{u_envelopes_qdir}, respectively.

The method can also be adapted to the rank envelope test, which requires the estimation of
the distribution of the extreme ranks $R_1$ under the composite hypothesis based on further simulations
and determining the empirical $\alpha$-quantile of this distribution.
A detailed algorithm for this procedure is given in Appendix \ref{sec:adjenv}.
This gives the critical rank $k_{\alpha^*}$ used to construct the \emph{adjusted rank envelope}, which is
limited by the curves $T_\low^{(k_\alpha^*)}$ and $T_\upp^{(k_\alpha^*)}$.
The test based on the adjusted rank envelope can be safely regarded as a test at the level $\alpha$, where
the null hypothesis is clearly rejected if the data function $T_1$ goes strictly outside the envelope,
and if $T_\low^{(k_\alpha^*)} < T_1 < T_\upp^{(k_\alpha^*)}$ for all $r \in I$, there is no evidence to reject the null hypothesis.

Note that the adjustment without further refinements requires $s^2$ simulations to achieve the same precision as an uncorrected plug-in Monte Carlo test based on $s$ simulations. Therefore, it is not very practical for the rank envelope test, where $s$ has to be large.

To reduce the running time of the algorithm for the rank envelope test, it is possible to approximate $\alpha^*$ from the $\alpha$-quantile of the additional $s$ $p_N$-values of the rank counts test using
a lower number of simulations $s_2 < s$, and then to find $k_{\alpha^*}$ by plugging $\alpha^*$ in  \eqref{eq:kalpha}.
This approximation is applied to an example case in Section \ref{sec:ex_tree_data} and it is compared with the ``no approximation'' case, although examining the quality of the approximation in various situations will be the subject of a future study.

\section{Simulation study}\label{sec:simstudy}

We studied the performance of the various tests listed in Table \ref{Table:testnames}.
First, we studied the power of the tests and the effect of the choice of the interval $I$ on the power under different scenarios (Section \ref{sec:simstudy_I}).
In this study, the CSR null hypothesis was tested against simulated regular and clustered point patterns (see Section \ref{sec:ppmodels} for description of the models).
As functional test statistics, we used non-parametric estimators of the $L$-function
\citep{Ripley1976, Ripley1977, Besag1977} (with translational edge correction) and the $J$-function \citep{LieshoutBaddeley1996} (with no edge correction), see Appendix \ref{sec:summary_functions} for further details.
Both functions are often used for detecting the clustering or regularity of point patterns.

Second, we considered practical composite hypothesis testing with plug-in parameter estimates using various point process models (Section \ref{sec:simstudy_composite_hypothesis}). In particular, we studied how the parameter estimation step affects the empirical type I error probabilities.
Note that a solution to correct the significance level of composite hypothesis tests is discussed in Section \ref{sec:composite} and applied to real data in Sections \ref{sec:ex_tree_data} and \ref{sec:ex_fallen_trees}.
We also conducted a small simulation study on the power of the tests in the presence of inhomogeneity in the tested/simulated models (Section \ref{sec:simstudy_inhomog}).

For each test in these experiments we used $s=1999$ simulations.
The rejection rates were calculated for the tests based on the $p$-value (rejection if $p \leq \alpha$), using the unique $p$-value obtained by rank count ordering for the rank envelope test.
In simulations not shown here, the rank count test based on $s=1999$ simulations yielded virtually the same proportion of rejections as the test based on $s=2499$ simulations (the number which we recommended for testing by the rank envelope test).

\begin{table}
\caption{\label{Table:testnames}Different tests and their abbreviated names (MAD = maximum absolute difference).}
\centering
\fbox{%
\begin{tabular}{ll}
Test & Short name \\\toprule
rank envelope test (rank count) & rank \\
(unscaled) MAD test \eqref{Dinfty_classical} & max \\
studentized MAD test \eqref{Dinfty_st}  & max\,|\,st \\
directional quantile MAD test \eqref{Dinfty_qdir} & max\,|\,qdir \\
integral deviation test \eqref{DL2_classical} & int \\
studentized integral deviation test & int\,|\,st  \\
directional quantile integral deviation test & int\,|\,qdir \\
modified half-region depth & MHRD  \\
modified band depth & MBD \\
\end{tabular}}
\end{table}

\subsection{Point process models}\label{sec:ppmodels}

Formally, a point pattern is a finite set $\x = \{x_1, \dots, x_n\}$ of
locations of $n$ objects in a region $D$, which is a compact
subset of $\mathbb{R}^d$.
We want to test the hypothesis $H_0$ that the observed point pattern $\x_W$
is a realization of a simple point process $X$ in a window $W \subseteq D$.
In the simulations, the window is a unit square.

The model for CSR is the Poisson
process with intensity $\lambda$.
The CSR hypothesis can be considered simple if the number of points is fixed
and the binomial point process is simulated \citep[see e.g.][]{Diggle2013}, which are the case in this study.
The other point process models used in the following experiment
are the regular Strauss process and the Mat{\'e}rn cluster processes explained in the following.

\begin{itemize}
\item {\em Strauss process}. The parameters of the pairwise interaction
point process Strauss($\beta$, $\gamma$, $R$) are the first-order term $\beta>0$,
the interaction radius $R>0$ and $0\leq \gamma\leq 1$, which controls the strength of interaction.
To obtain point patterns that behave like patterns derived from a stationary model,
the point patterns in $W$ were obtained by simulating the process in an extended window
and taking the sub-point patterns in $W$ as samples.
The simulations in the present study were performed in a unit square
(extending by $0.25$ in each direction) using a perfect simulation algorithm
implemented in the R library spatstat \citep{BerthelsenMoller2003, BaddeleyTurner2005}.

\item {\em Mat{\'e}rn cluster process}. Mat{\'e}rn type cluster processes are
constructed in two steps. First, a ``parent'' point process is
generated and clusters of ``daughter'' points are then
formed around the parent points.
In the classical {\em Mat{\'e}rn cluster process}
MatClust($\lambda_p$, $R_d$, $\mu_d$), the parent points form the
Poisson process with intensity $\lambda_p$ and the daughter points
are distributed uniformly in discs of radius $R_d$ centred at the
parent points, where their numbers follow a Poisson distribution with
mean $\mu_d$.

\item The {\em non-overlapping Mat{\'e}rn cluster process} NoOMat\-Clust($\lambda_p$, $R_d$, $\mu_d$, $R$)  differs from the Mat{\'e}rn cluster process only in that the parent points follow the Strauss($\lambda_p$, $0$, $R$) process, i.e.,\ a {\em hard core process} with minimum inter-point distance $R$. The condition $R > R_d$ ensures that clusters are non-overlapping.
\item The {\em mixed Mat{\'e}rn cluster process} (MixMatClust) is a superposition of two Mat{\'e}rn cluster processes.
\end{itemize}

The parent points of the Mat{\'e}rn, non-overlapping Mat{\'e}rn and mixed Mat{\'e}rn processes were simulated in an extended window $[-R_d,1+R_d]\times[-R_d,1+R_d]$, because a daughter point inside $W$ could have its parent point outside $W$. 
Further, the parent points of the non-overlapping Mat{\'e}rn cluster process were first simulated in an larger extended window (extending by 0.25 in each direction) similarly as the Strauss processes of this paper.
The final daughter point patterns were clipped to $W=[0,1]\times[0,1]$.

Figure \ref{fig:intervalIstudy_expatterns} shows typical realizations from the Strauss and Mat{\'e}rn cluster process models used in the simulation experiment of Section \ref{sec:simstudy_I}, and from the non-overlapping and mixed Mat{\'e}rn cluster processes used in the simulation study of Section \ref{sec:simstudy_composite_hypothesis}.

\begin{figure}[htbp]
\centering
\makebox{\includegraphics[width=0.35\textwidth]{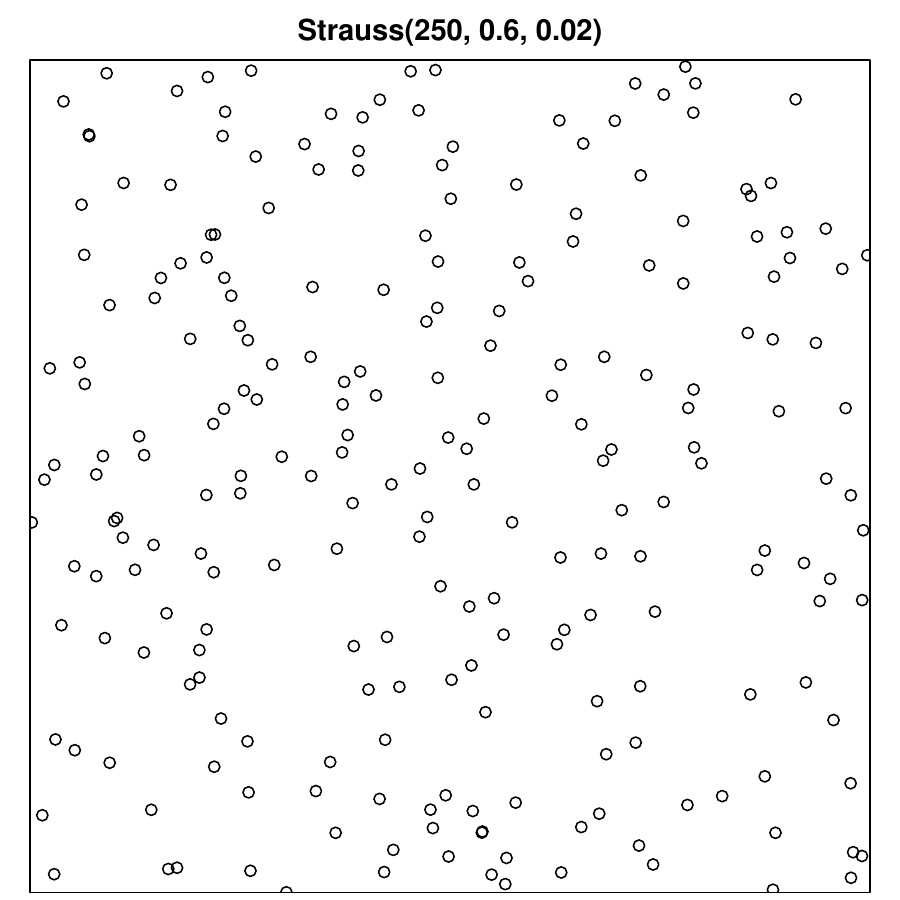} \hfill \includegraphics[width=0.35\textwidth]{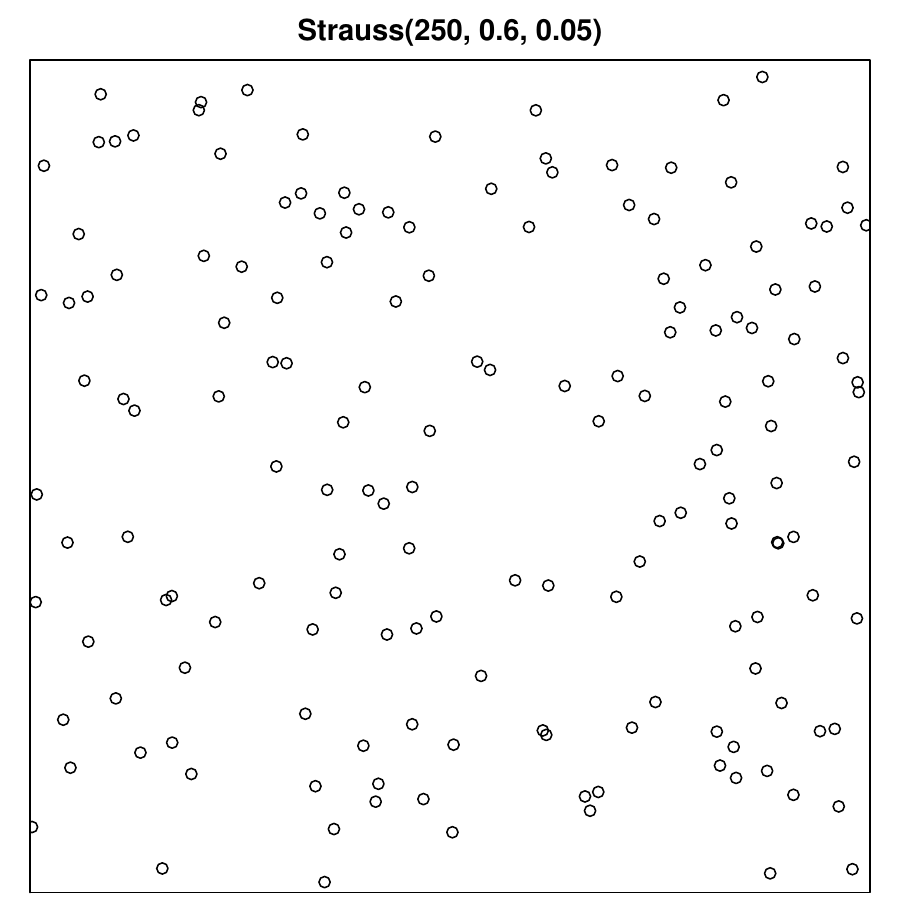}}
\makebox{\includegraphics[width=0.35\textwidth]{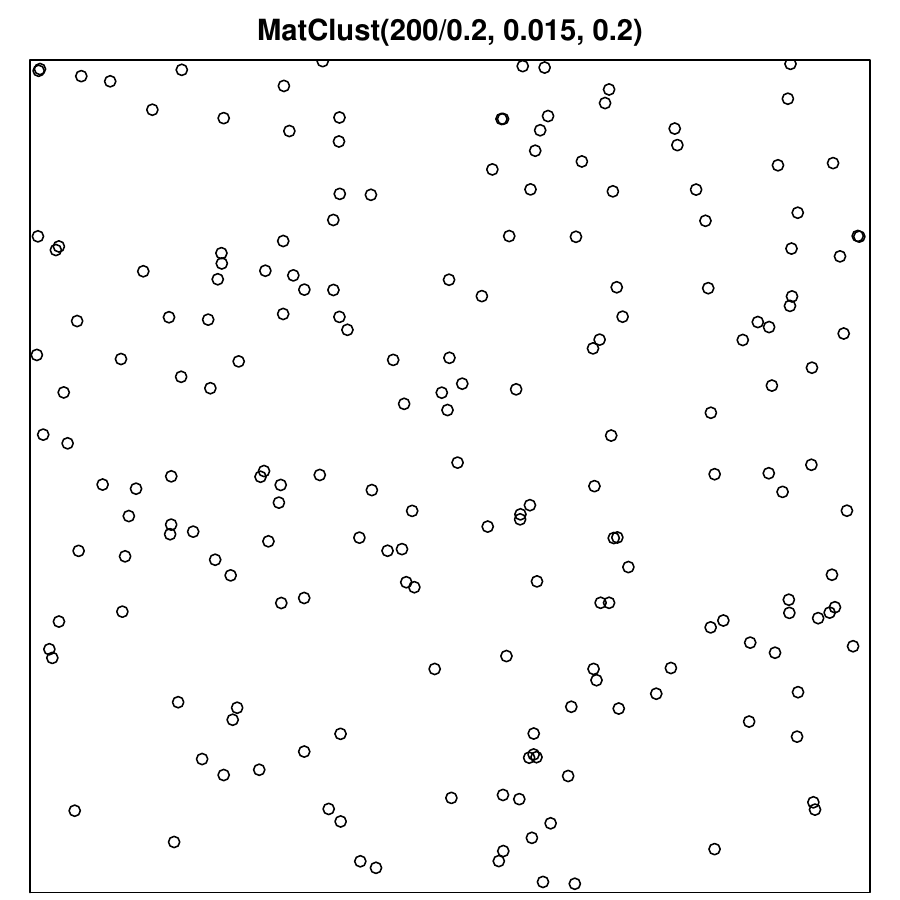} \hfill \includegraphics[width=0.35\textwidth]{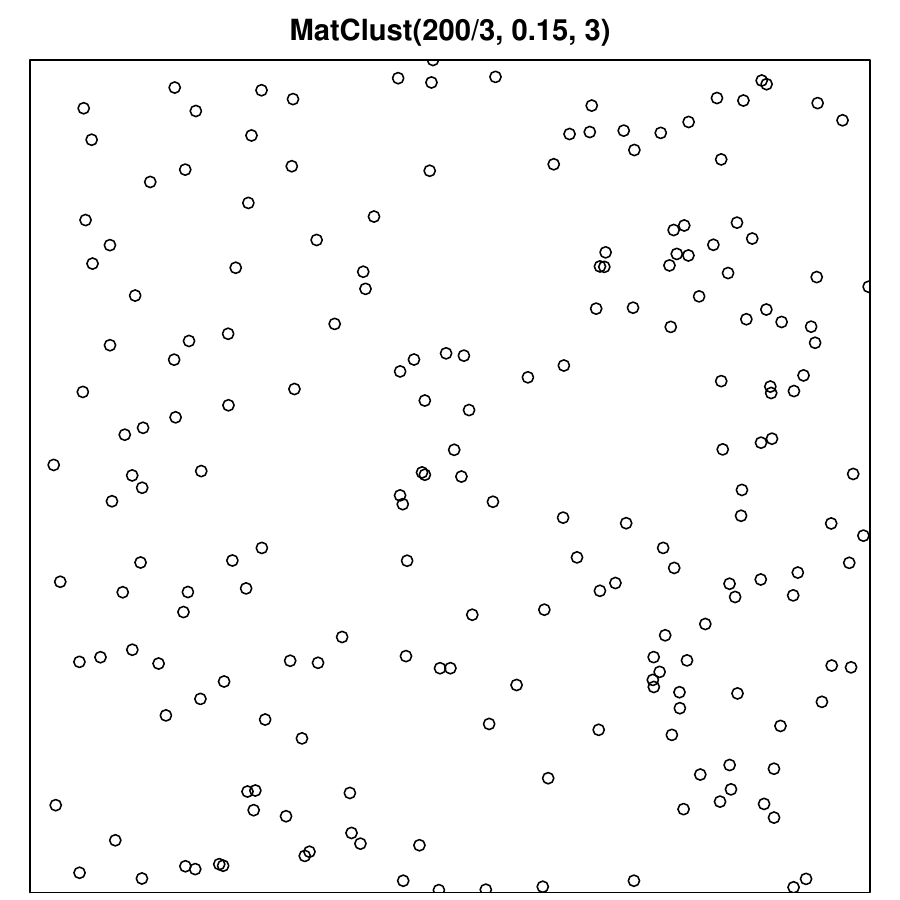}}
\makebox{\includegraphics[width=0.35\textwidth]{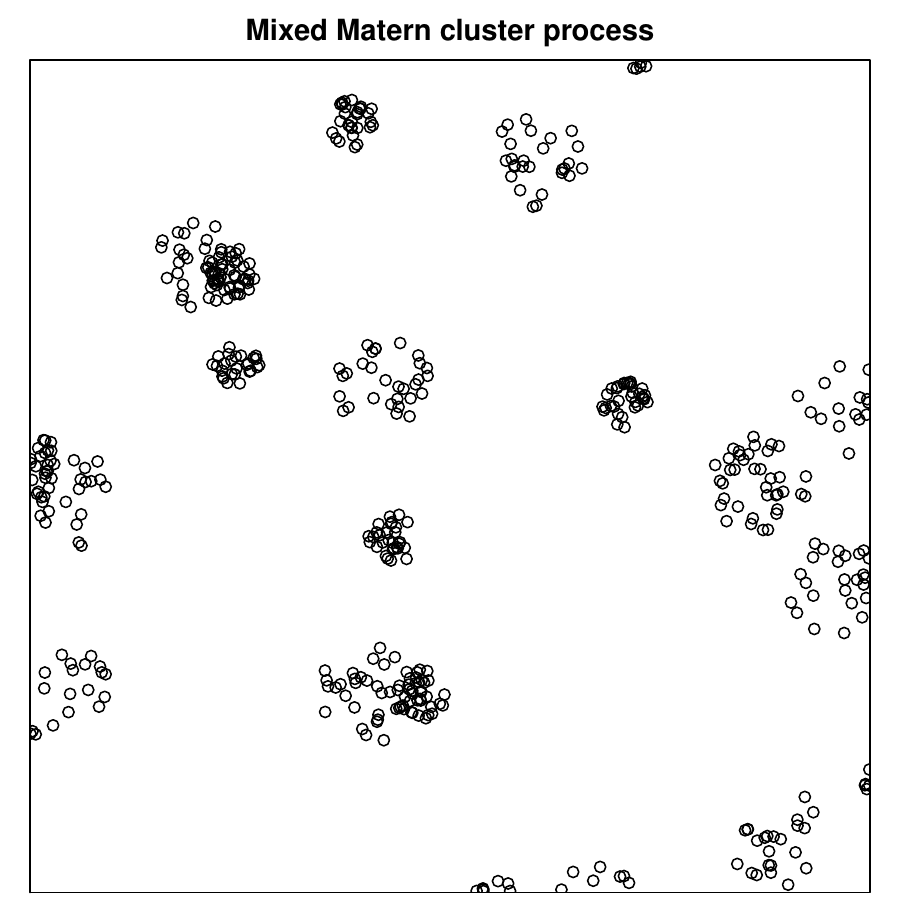} \hfill \includegraphics[width=0.35\textwidth]{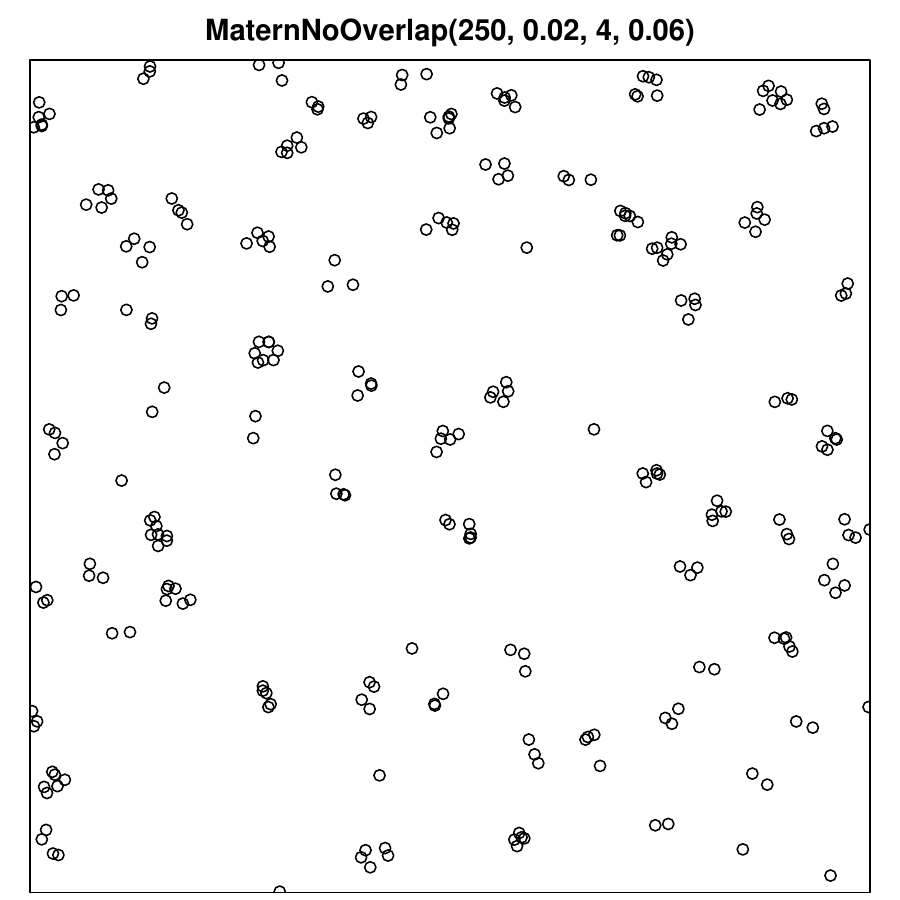}}
\caption{\label{fig:intervalIstudy_expatterns}Example patterns of the repulsive Strauss(250, 0.6, 0.02) and Strauss(250, 0.6, 0.05) processes, and the clustered MatClust(200/0.2, 0.015, 0.2) and MatClust(200/3, 0.15, 3) processes used in the simulation study of Section \ref{sec:simstudy_I} (see description of the processes there) as well as examples of the mixed and non-overlapping Mat{\'e}rn cluster processes used in simulation study of Section \ref{sec:simstudy_composite_hypothesis}. The mixed Mat{\'e}rn cluster process is a superposition of the MatClust(10, 0.06, 30) and MatClust(10, 0.03, 30) processes.}
\end{figure}

\subsection{Simple hypothesis testing: power comparison and the choice of the interval $I$}\label{sec:simstudy_I}

We started our power comparison of the tests in Table \ref{Table:testnames} using the simple CSR hypothesis.
The power of tests depends on the choice of the interval $I=[r_{\min},r_{\max}]$, so
we fixed $r_{\min}=0.005$ and estimated the power with respect to the increasing upper bound $r_{\max}$
in two different scenarios.
In the first scenario, small distances are expected to lead to the rejection of the null hypothesis, and thus the tests should be able to reject the null hypothesis for a small upper bound. In the second scenario, intermediate distances are expected to lead to rejection and the tests should not be capable of rejection for a small upper bound.
We studied both scenarios for repulsive (Strauss(250, 0.6, 0.02) and Strauss(250, 0.6, 0.05) with interaction radii 0.02 and 0.05) and clustered (MatClust(200/0.2, 0.015, 0.2) and MatClust(200/3, 0.15, 3)) patterns (see Section \ref{sec:ppmodels} for the specifications of the models).
The expected number of points for the Mat{\'e}rn processes was 200, while it was approximately 220 and 150 for the two Strauss processes, respectively. The first Mat{\'e}rn process had tiny clusters (disc radius $R_d=0.015$) with expected number of points 0.2, i.e.\ only part of the realized ``clusters'' had $\geq 2$ points, while in the second process, the expected number of points in a cluster with a much larger disc radius $R_d=0.15$ was 3.
Typical realizations of the processes are shown in Figure \ref{fig:intervalIstudy_expatterns}.

For each of the four models, the powers were estimated from 100 realizations.
The results are shown in Figures \ref{fig:choiceL} and \ref{fig:choiceJ} for $T(r)=\hat{L}(r)$ and $T(r)=\hat{J}(r)$, respectively.
The power curves for the studentized deviation tests and for the MBD test are omitted from the figures to simplify the presentation because the performance of these tests was very similar to that of the directional quantile tests and the MHRD test, respectively.

\begin{figure}[htbp]
\centering
\makebox{\includegraphics[scale=0.55]{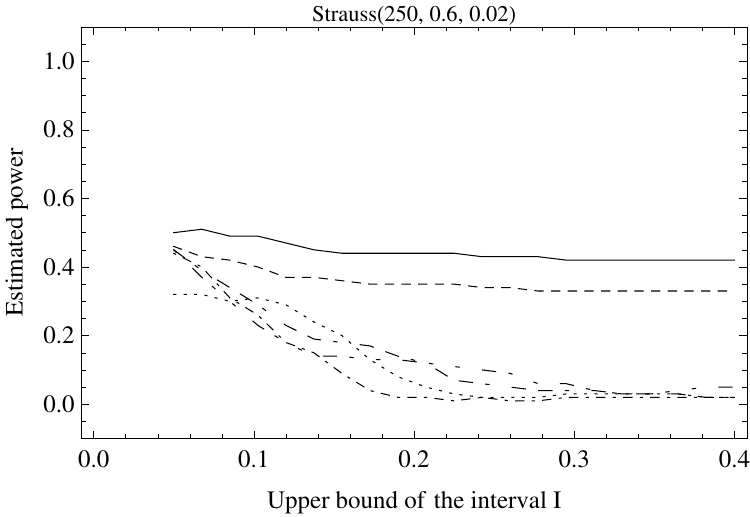} \hfill \includegraphics[scale=0.55]{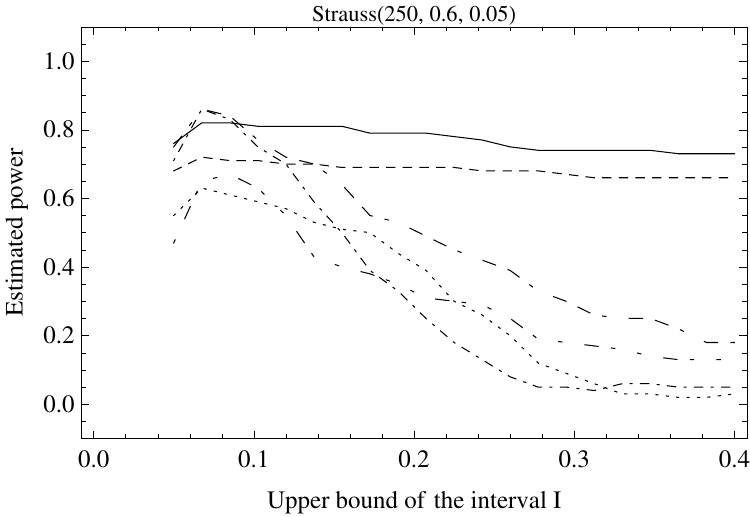}}
\makebox{\includegraphics[scale=0.55]{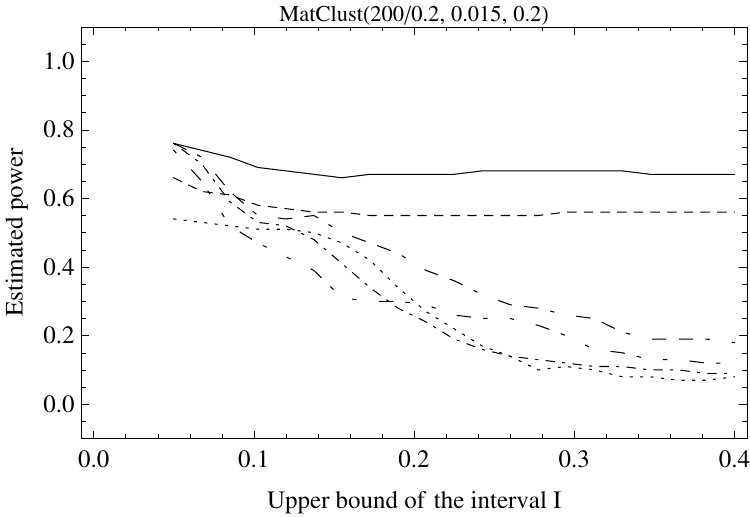} \hfill \includegraphics[scale=0.55]{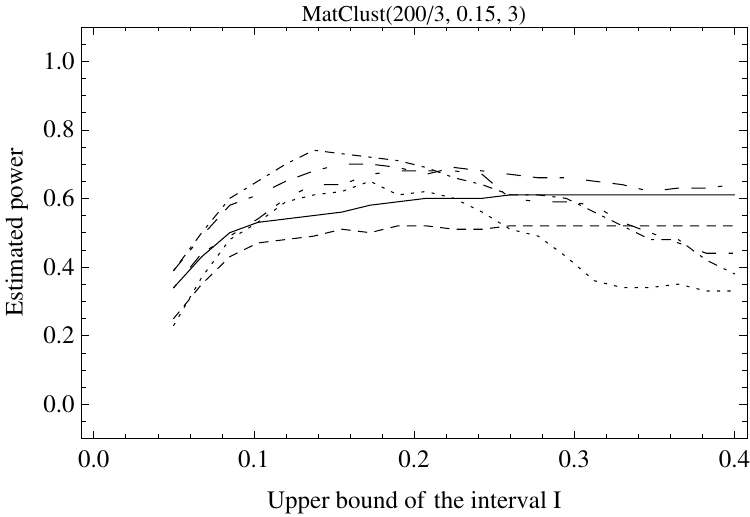}}
\caption{\label{fig:choiceL} The power curves for rank (solid), max\,|\,qdir  (dashed), max (dotted), int (dot dashed), int\,|\,qdir  (dot dash dashed) and MHRD (dot dot dashed) tests with respect to the growing upper limit of the interval $I$  from 0.05 to 0.4. All of the simulated models (indicated above the plots) were tested for CSR with $T(r)=\hat{L}(r)$ and $s=1999$. The full names of the different tests are given in Table \ref{Table:testnames}.}
\end{figure}

\begin{figure}[htbp]
\centering
\makebox{\includegraphics[scale=0.55]{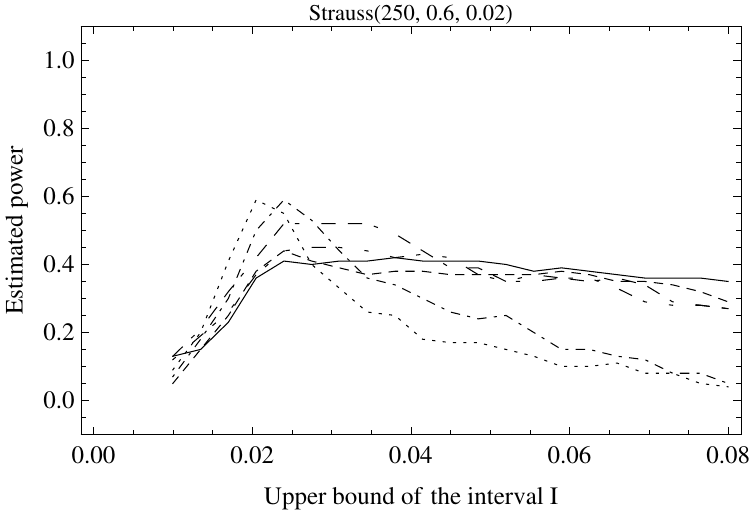} \hfill \includegraphics[scale=0.55]{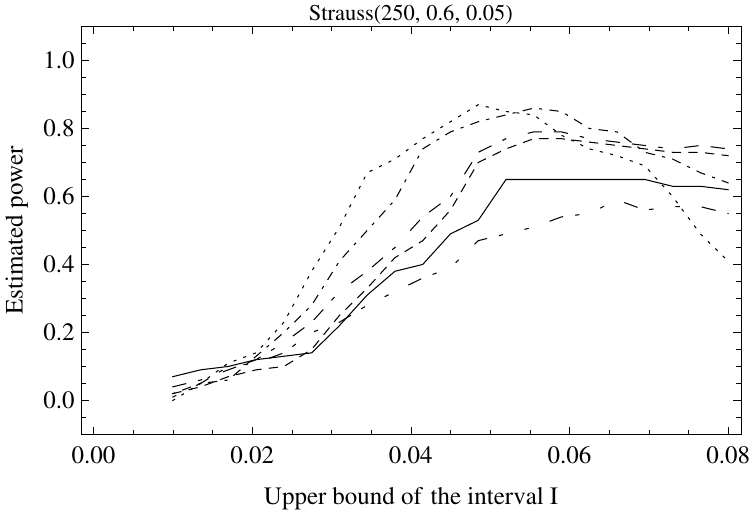}}
\makebox{\includegraphics[scale=0.55]{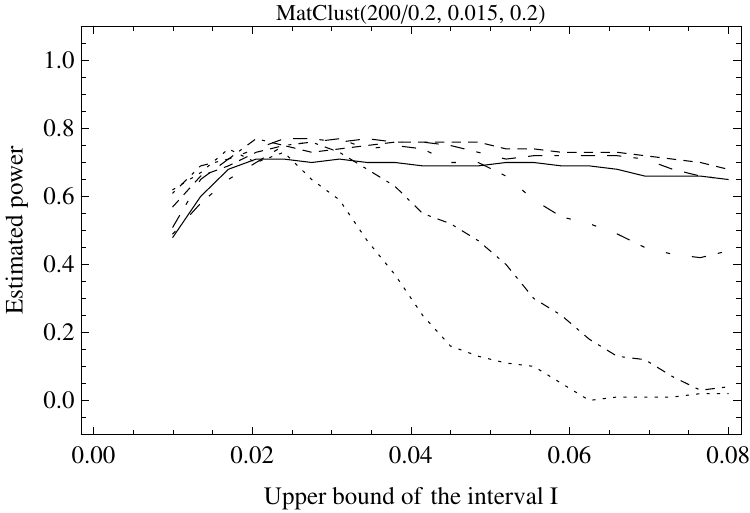} \hfill \includegraphics[scale=0.55]{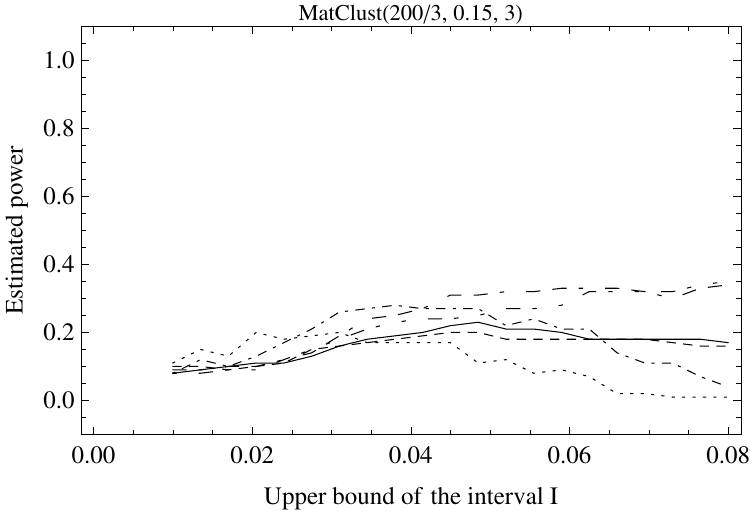}}
\caption{\label{fig:choiceJ} The power curves for rank (solid), max\,|\,qdir  (dashed), max (dotted), int (dot dashed), int\,|\,qdir  (dot dash dashed) and MHRD (dot dot dashed) tests with respect to the growing upper limit of the interval $I$ from 0.01 to 0.08. All of the simulated models (indicated above the plots) were tested for CSR with $T(r)=\hat{J}(r)$ and $s=1999$. The full names of the different tests are given in Table \ref{Table:testnames}.}
\end{figure}

The results show that the rank and scaled (directional quantile and studentized) MAD tests were the most robust tests with respect to the choice of the interval $I$. They also exhibited good power.
In some cases, the scaled integral deviation tests were slightly more powerful than
the rank and scaled MAD tests. This generally happens when a long interval of distances within $I$ is responsible for rejection.
The classical MAD and integral deviation tests, as well as the MHRD and MBD tests, were generally less powerful and less robust than the rank and scaled MAD tests.

Note that for the Strauss processes, the (unscaled) MAD test with $T(r)=\hat{J}(r)$ had the highest power for $r_{\max}\leq R$, where $R$ is the interaction radius of the process (see Figure \ref{fig:choiceJ}). 
This was because the largest deviation of $T_1(r)$ from $T_0(r)$ occurred for the largest distances on $I$, where the test without weights thus gave the greatest importance.

Thus, for the rank and scaled MAD tests, we can conclude that the choice of the interval $I$
plays only a minor role.
Therefore, the interval $I$ can be taken to include the whole domain of interesting distances without fear of losing power.  In contrast, the scaled integral deviation tests did not have such a robust performance.

We also conducted a simulation study to compare the performance of our proposed tests with the tests proposed by \citet{HoChiu2006}, who advocated the use of adapted distance-dependent intensity estimators for the $K$-function.
We obtained comparable results to those obtained by \citet{HoChiu2006} using our simple scaled MAD tests (rank count ordering required 999 simulations instead of 99 to obtain comparable results).

\subsection{Plug-in composite hypothesis testing}\label{sec:simstudy_composite_hypothesis}

In the study reported below, we used the intervals $I=[0.005, 0.2]$ and $I=[0.005, 0.05]$ for the functional test statistics $\hat{L}(r)$ and $\hat{J}(r)$, respectively. 

\subsubsection{Type I error probabilities}

In the case of simple hypotheses, we confirmed empirically that the estimates of the type I error
were close to the nominal level.
Table \ref{Table:results_signlevel_known} shows the empirical type I error probabilities estimated from $N=1000$ realizations of the three simulated processes under the condition that the true (simulation) model was equal to the null model with fixed known parameters.
For $\alpha=0.05$, the proportions of rejections should be in the interval $(0.037, 0.064)$
with a probability of 0.95 (given by the 2.5\% and 97.5\% quantiles of the binomial distribution with parameters
1000 and 0.05), so it can be concluded that all of the tests had correct empirical type I error probabilities.

\begin{table}
\caption{\label{Table:results_signlevel_known}Proportions of rejections (among $N=1000$ repetitions of the test) when the null model was equal to the true (simulation) model (with known parameter values), and the null hypothesis was rejected if $p\leq 0.05$. The full names of the different tests in the columns are given in Table \ref{Table:testnames}.}
\centering
\resizebox{\columnwidth}{!}{
\fbox{%
 \begin{tabular}{ll|c|ccc|ccc|cc}
 True model &  $T(r)$               &   rank  & max  & max\,|\,st & max\,|\,qdir & int & int\,|\,st & int\,|\,qdir & MBD & MHRD \\\toprule
 Binomial(200)           &  $\hat{L}(r)$  & 0.056 & 0.062 & 0.053 & 0.056 & 0.059 & 0.064 & 0.059 & 0.059 & 0.060 \\
                         &  $\hat{J}(r)$  & 0.046 & 0.046 & 0.048 & 0.047 & 0.056 & 0.049 & 0.048 & 0.046 & 0.048 \\ \midrule
 Strauss(350, 0.4, 0.03) &  $\hat{L}(r)$   & 0.053 & 0.056 & 0.057 & 0.056 & 0.054 & 0.048 & 0.054 & 0.057 & 0.062 \\
                         &  $\hat{J}(r)$ & 0.053 & 0.046 & 0.045 & 0.053 & 0.046 & 0.052 & 0.052 & 0.054 & 0.052 \\ \midrule
 MatClust(50, 0.06, 4) & $\hat{L}(r)$     & 0.050 & 0.048 & 0.053 & 0.055 & 0.046 & 0.044 & 0.044 & 0.042 & 0.043 \\
                       &  $\hat{J}(r)$  & 0.051 & 0.058 & 0.051 & 0.049 & 0.053 & 0.058 & 0.051 & 0.053 & 0.053 \\
\end{tabular}}}
\end{table}

In practice, the parameter values of a model are not known and a goodness-of-fit test is often performed according to the following steps.
\begin{enumerate}[(a)]
 \item Estimate the parameters of the null model from the data pattern observed in $W$.
 \item Generate $s$ samples of the null model in $W$ with the estimated parameters.
 \item Estimate $T(r)$ for the data pattern and each simulation.
 \item Calculate the $p$-value based on the chosen ordering.
 \item Reject $H_0$ if $p \leq \alpha$.
\end{enumerate}
Thus, we checked how the parameter estimation step affects the type I error probabilities in our tests for Strauss and Mat{\'e}rn cluster processes.
We estimated the parameters of the models using the R library spatstat (1.31-2) \citep{BaddeleyTurner2005}. Estimates for the parameters of the Mat{\'e}rn cluster process were obtained from the minimum contrast estimation based on the pair correlation function (the non-cumulative counterpart of the $L$-function), whereas for the Strauss process, we used the maximum pseudo-likelihood (MPL) method \citep{BaddeleyTurner2000, DiggleEtal1994} (default method in spatstat) and the logistic likelihood method \citep{BaddeleyEtAl2014a}.

Table \ref{Table:results_signlevel_fitted} shows the proportions of rejections for the fitted Strauss and Mat{\'e}rn cluster models with nuisance parameters.
In general, the tests were conservative. The conservativeness was particularly severe for the Mat{\'e}rn cluster process in the tests based on $\hat{L}(r)$, which was due to the strong interdependency between the fitting method and the test statistic.
There were two exceptions to this conservativeness.
\begin{enumerate}[(a)]
 \item The empirical levels of the tests based on $\hat{J}(r)$ for the Mat{\'e}rn cluster process were close to the nominal level, except for the MBD and MHRD tests.
 \item The tests based on $\hat{J}(r)$ for the Strauss process where the parameters were estimated by
the MPL method with the default (in spatstat) auxiliary parameters
tended to be liberal, particularly for the classical maximum and integral deviation tests.
\end{enumerate}
The first exception appears to have occurred because the estimation method and $\hat{J}(r)$ are only loosely related. The second exception appears to have been due to a poor estimation method (see Table \ref{Table:Strauss_parameter_estimates}) and its interaction with $\hat{J}(r)$.
We conclude that a poor estimation method can mask the conservativeness of a test and lead to less conservative or even liberal testing. We note that the mean (standard deviation) of the parameter estimates for the Mat{\'e}rn cluster process appeared to be adequate: $\hat{\lambda}_b=57.05$ (21.42), $\hat{R}_d=0.056$ (0.011), $\hat{\mu}_d=3.93$ (1.62).

\begin{table}
\caption{\label{Table:results_signlevel_fitted}Proportions of rejections of the null models (among $N=1000$ repetitions of the test at $\alpha=0.05$) using the estimated parameters when the null model was fitted to a realization of the true (simulation) model. In the case of the Strauss process, the interaction radius was fixed to $R=0.03$ and the other parameters were estimated based on the maximum pseudo-likelihood (MPL) method or the logistic likelihood (LL) method. The full names of the different tests in the columns are given in Table \ref{Table:testnames}.}
\centering
\resizebox{\columnwidth}{!}{
\fbox{%
\begin{tabular}{ll|c|ccc|ccc|cc}
Null model &\multirow{2}{*}{$T(r)$}   & \multirow{2}{*}{rank}  & \multirow{2}{*}{max}  & \multirow{2}{*}{max\,|\,st} & \multirow{2}{*}{max\,|\,qdir} & \multirow{2}{*}{int} & \multirow{2}{*}{int\,|\,st} & \multirow{2}{*}{int\,|\,qdir} &
\multirow{2}{*}{MBD} & \multirow{2}{*}{MHRD}
\\ / True model &&&&&&&&&&
\\\toprule
Strauss($R=0.03$)$^{(\text{MPL})}$ &  $\hat{L}(r)$  & 0.026 & 0.044 & 0.040 & 0.039 & 0.027 & 0.025 & 0.025 & 0.027 & 0.033 \\
/ Strauss(350, 0.4, 0.03)  &  $\hat{J}(r)$  & 0.066 & 0.092 & 0.049 & 0.060 & 0.089 & 0.068 & 0.063 & 0.076 & 0.076 \\ \midrule
 Strauss($R=0.03$)$^{(\text{LL})}$ &  $\hat{L}(r)$  & 0.027 & 0.034 & 0.034 & 0.039 & 0.037 & 0.040 & 0.041 & 0.036 & 0.029 \\
/ Strauss(350, 0.4, 0.03) &                          $\hat{J}(r)$  & 0.010 & 0.028 & 0.024 & 0.016 & 0.017 & 0.005 & 0.005 & 0.002 & 0.002 \\  \midrule
 MatClust & $\hat{L}(r)$    & 0.007 & 0.009 & 0.020 & 0.020 & 0 & 0.001 & 0 & 0 & 0 \\ 
/ MatClust(50, 0.06, 4)   &  $\hat{J}(r)$  & 0.048 & 0.056 & 0.062 & 0.052 & 0.053 & 0.056 & 0.044 & 0.023 & 0.015 \\
\end{tabular}}}
\end{table}

\begin{table}
\caption{\label{Table:Strauss_parameter_estimates}Mean (sd) of the parameter estimates for the Strauss($\beta=350$, $\gamma=0.4$, $R=0.03$) process with the maximum pseudo-likelihood (MPL) and logistic likelihood (LL) methods based on $N=1000$ realizations. The interaction radius $R=0.03$ was fixed.}
\centering
\fbox{%
\begin{tabular}{l|cc}
Estimation method & $\beta$ & $\gamma$ \\\toprule
MPL               & 307.05 (30.69)  & 0.49 (0.095) \\
LL                & 346.65 (35.92)  & 0.42 (0.077) \\
\end{tabular}}
\end{table}

\subsubsection{Comparison of the rejection rates}\label{sec:simstudy_power}

In practice, the tests for composite hypotheses are often used with plug-in parameter values (and conservativeness accepted), so we compared the proportions of rejections with the different plug-in tests for the three cases listed in the first two columns of Table \ref{Table:results_composite_hypotheses}.
We recall that the tests are conservative, except for the Mat{\'e}rn cluster process tested with $\hat{J}(r)$.
Nevertheless, a power comparison is meaningful because the proportions
in the lines of Table \ref{Table:results_signlevel_fitted} are statistically indistinguishable,
except for the Strauss case with MPL estimation with $\hat{J}(r)$, as well as a few cases for MBD and MHRD tests.
Similar to the power comparison in Section \ref{sec:simstudy_I},
the rank envelope and scaled MAD tests clearly had higher rejection rates than the classical deviation tests, or the tests based on MBD and MHRD.

\begin{table}
\caption{\label{Table:results_composite_hypotheses}Proportions of rejections of the null models (among $N=1000$ repetitions of the test at $\alpha=0.05$) when the simulated pattern was used first to fit the null model and then for testing. The logistic likelihood (LL) method was used for the Strauss process. The mixed Mat{\'e}rn cluster process is a superposition of the MatClust(10, 0.06, 30) and MatClust(10, 0.03, 30) processes, and the non-overlapping process is MaternNoOverlap(250, 0.02, 4, 0.06). The full names of the different tests in the columns are given in Table \ref{Table:testnames}.}
\centering
\resizebox{\columnwidth}{!}{
\fbox{%
\begin{tabular}{ll|c|ccc|ccc|cc}
Null model  &\multirow{2}{*}{$T(r)$}    & \multirow{2}{*}{rank}  & \multirow{2}{*}{max}  & \multirow{2}{*}{max\,|\,st} & \multirow{2}{*}{max\,|\,qdir} & \multirow{2}{*}{int} & \multirow{2}{*}{int\,|\,st} & \multirow{2}{*}{int\,|\,qdir} & \multirow{2}{*}{MBD} & \multirow{2}{*}{MHRD}
\\
/ True model &&&&&&&&&&
        \\\toprule
Strauss(R=0.02)$^{(\text{LL})}$     & $\hat{L}(r)$ & 0.789 & 0.362 & 0.788 & 0.783 & 0.069 & 0.161 & 0.175 & 0.098 & 0.111 \\
/ Strauss(350, 0.4, 0.03)                       & $\hat{J}(r)$  & 0.732 & 0.378 & 0.731 & 0.780 & 0.593 & 0.696 & 0.662 & 0.290 & 0.073 \\ \midrule
MatClust         & $\hat{L}(r)$   & 0.005 & 0.011 & 0.004 & 0.006 & 0.001 & 0 & 0 & 0 & 0 \\
/ MixMatClust   & $\hat{J}(r)$  &0.942 & 0.557 & 0.952 & 0.956 & 0.842 & 0.932 & 0.941 & 0.292 & 0.288 \\ \midrule
MatClust    & $\hat{L}(r)$  & 0.564 & 0.063 & 0.257 & 0.481 & 0.100 & 0.138 & 0.234 & 0.107 & 0.049 \\
/ NoOMatClust  & $\hat{J}(r)$  & 0.161 & 0.008 & 0.044 & 0.106 & 0.053 & 0.107 & 0.182 & 0.089 & 0.017 \\
\end{tabular}}}
\end{table}

\subsection{Simulation study with inhomogeneous models}\label{sec:simstudy_inhomog}

We also studied the performance of different tests to detect interaction between points in the presence of inhomogeneity. For this purpose, we generated ``data'' patterns by means of the Strauss($\beta$, 0.6, 0.03) process with homogeneous and inhomogeneous first-order terms:
\begin{enumerate}[(a)]
 \item the stationary process with $\beta=250$, and
 \item $\beta(x,y)=\exp(\beta_0+\beta_1 x+ \beta_2 y)$ with $(\beta_0,\beta_1,\beta_2)=(5,0.5,0.5)$, leading to increasing intensity with respect to the coordinates.
\end{enumerate}
Both processes had approximately the same expected number of points, i.e., 200 in the unit square window.
The null model in the case of (a) was CSR, while in the case of (b), the null hypothesis was that data are compatible with the inhomogeneous Poisson process where the intensity is of the same form as the $\beta$-term, i.e., the form of inhomogeneity was assumed to be known, but the parameters unknown.
In the case of (a), we used the $L$- and $J$-functions described in Appendix \ref{sec:summary_functions},
whereas in the case of (b), we employed the inhomogeneous versions of these functions \citep{BaddeleyEtal2000, Lieshout2011} using the estimated parametric intensity. 
In the latter case (b), we explored two options, namely the cases where
\begin{enumerate}
\renewcommand{\theenumi}{(\roman{enumi})}
\item the parametric intensity was estimated for the ``data'' pattern and the test statistic was the estimator of the $L_{\inhom}$- or $J_{\inhom}$-function with this estimated intensity, 
\item the parametric intensity was estimated for the ``data'' pattern and reestimated for each simulated pattern $s=2,\dots,s+1$ in the test, i.e., the test statistic was the estimator of $L_{\inhom}$ or $J_{\inhom}$ with intensity to be estimated parametrically from the point pattern.
\end{enumerate}
The case (i) seem the preferred choice because, in this case, the parameters of the intensity need to be estimated only once. In our study, we estimated the parameters $\beta_0, \beta_1, \beta_2$ of the exponential intensity functions by the standard maximum likelihood estimation using the R library spatstat \citep{BaddeleyTurner2005}. This is a practical estimation method that is commonly applied to estimate covariate effects when using real data. 
Furthermore, we took $I=[0.005, 0.2]$ for the $L$-functions and $I=[0.005, 0.05]$ for $J$-functions, as in the homogeneous study described in Section \ref{sec:simstudy_composite_hypothesis}, and, for estimation of empirical type I error probabilities and rejection rates, we generated $N=500$ realizations from each model considered and repeated tests for each generated pattern.

First, to see how the empirical type I error probabilities were affected by estimation of the parameters of the intensity, we investigated the empirical type I error probabilities for the inhomogeneous Poisson process with increasing intensity in the case of the functional test statistics described above. 
Table \ref{Table:results_inhom_typeIerror} shows the results, which we summarize for the cases (i) and (ii) as follows:
\begin{enumerate}
\renewcommand{\theenumi}{(\roman{enumi})}
\item The empirical type I error probabilities were quite close to the nominal level for $T(r) = \hat{J}_{\inhom}(r)$. For $T(r)=\hat{L}_{\inhom}(r)$, all of the other tests were extremely conservative, except for the rank and scaled MAD tests which were only somewhat conservative.
\item The empirical type I error probabilities were quite close to the nominal level for all tests and both test functions.
\end{enumerate}
Note that the proportions of rejections for $N=500$ should lie between 0.032 and 0.070 with probability 0.95.
Table \ref{Table:results_inhom_typeIerror} presents the empirical type I error probabilities also for the homogeneous Poisson process (CSR) case, where the simulations were generated from the Poisson process with estimated intensity (instead of fixing the number of points). These probabilities were close to the nominal level, indicating that 
the type I error probabilities were not noticeably affected by the varying number of points. See, however, \citet{BaddeleyEtal2014} where opposite results were obtained for testing a Poisson point process with a low intensity.

We also tested the null hypothesis that there is no interaction against the Strauss processes with the inhomogeneity parts specified by the $\beta$-terms.
We simulated the inhomogeneous Strauss model based on a Metropolis-Hastings algorithm, using the R library spatstat. Table \ref{Table:results_inhom} shows the rejection rates when the ``data'' patterns were generated by the Strauss processes.
It can be seen that the rank and scaled MAD tests were the most powerful of all the tests studied for $\hat{L}_{\inhom}$. For $\hat{J}_{\inhom}$, the scaled integral deviation tests had slightly higher rejection rates.
For $\hat{J}_{\inhom}$ in the case of (i) and for both $\hat{L}_{\inhom}$ and $\hat{J}_{\inhom}$ in the case of (ii), the rejection rates for the inhomogeneous Strauss processes were almost the same as those for the homogeneous Strauss process. For $\hat{L}_{\inhom}$ in the case of (i), all of the other tests, except the rank and scaled MAD tests, lost their ability to differentiate the Strauss inhibition processes from the ``no interaction'' case (note the conservativeness of the tests).

Thus, the results were rather dependent on the intensity used in the inhomogeneous functional test statistics, but the rank and scaled MAD tests were less affected by the choice of the intensity than the other tests.
We conclude that the better strategy appears to be to reestimate the intensity for each simulated pattern in the test (case (ii)), at least when $\hat{L}_{\inhom}$ is used as the functional test statistic, because then tests are likely to be only slightly conservative and this leads to more powerful tests.

\begin{table}
\caption{\label{Table:results_inhom_typeIerror}Proportions of rejections (among $N=500$ repetitions of the test at $\alpha=0.05$) of the null models of the homogeneous Poisson process (HPP) and inhomogeneous Poisson processes (IPP) using the estimated parameters. For the HPP, only the constant intensity $\lambda=250$  was estimated. For the IPP, the null model was fitted to a realization of the true (simulation) model (= a ``data'' pattern) with ``linear'' intensity $\exp(\beta_0+\beta_1 x+\beta_2 y)$ with $(\beta_0,\beta_1,\beta_2)=(5,0.5,0.5)$. The parametric form of the intensity was fixed and only the parameters $\beta_0,\beta_1,\beta_2$ were estimated. 
In the case (i) the intensity used in the estimator of $L_{\inhom}$ or $J_{\inhom}$ was fixed to that estimated from the ``data'' pattern, whereas in the case (ii), the intensity was reestimated for each simulation of the inhomogeneous Poisson process. The full names of the different tests in the columns are given in Table \ref{Table:testnames}.}
\centering
\resizebox{\columnwidth}{!}{
\fbox{%
\begin{tabular}{ll|c|ccc|ccc|cc}
True model & $T(r)$   & rank  & max  & max\,|\,st & max\,|\,qdir & int & int\,|\,st & int\,|\,qdir & MBD & MHRD
        \\\toprule
HPP(250) & $\hat{L}(r)$   & 0.047 & 0.052 & 0.053 & 0.053 & 0.056 & 0.057 & 0.048 & 0.053 & 0.056 \\ 
       & $\hat{J}(r)$   & 0.050 & 0.052 & 0.046 & 0.050 & 0.044 & 0.052 & 0.050 & 0.049 & 0.040 \\ \midrule
IPP(linear) & $\hat{L}_{\inhom}(r)^{(i)}$ & 0.026 & 0 & 0.018 & 0.026 & 0 & 0 & 0 & 0 & 0 \\ 
                 & $\hat{J}_{\inhom}(r)^{(i)}$ & 0.056 & 0.054 & 0.072 & 0.068 & 0.048 & 0.060 & 0.058 & 0.054 & 0.056 \\ \midrule
IPP(linear) & $\hat{L}_{\inhom}(r)^{(ii)}$ & 0.038 & 0.044 & 0.044 & 0.042 & 0.040 & 0.036 & 0.036 & 0.046 & 0.052 \\ 
                 & $\hat{J}_{\inhom}(r)^{(ii)}$ & 0.032 & 0.046 & 0.028 & 0.034 & 0.050 & 0.030 & 0.030 & 0.040 & 0.040 \\ 
\end{tabular}}}
\end{table}

\begin{table}
\caption{\label{Table:results_inhom}Proportions of rejections (among $N=500$ repetitions of the test at $\alpha=0.05$) of the null models when the ``data'' patterns were generated by the true (simulation) models with the following first-order terms: (a) constant 250, (b) ``linear'' $\exp(\beta_0+\beta_1 x+\beta_2 y)$ with $(\beta_0,\beta_1,\beta_2)=(5,0.5,0.5)$. The fitted inhomogeneous Poisson process (IPP) had the same parametric form as its intensity and only the parameters $\beta_0,\beta_1,\beta_2$ were estimated. In the case (i) the intensity used in the estimator of $L_{\inhom}$ or $J_{\inhom}$ was fixed to that estimated from the ``data'' pattern, whereas in the case (ii), the intensity was reestimated for each simulation of the inhomogeneous Poisson process. The full names of the different tests in the columns are given in Table \ref{Table:testnames}.}
\centering
\resizebox{\columnwidth}{!}{
\fbox{%
\begin{tabular}{ll|c|ccc|ccc|cc}
Null model  &\multirow{2}{*}{$T(r)$}    & \multirow{2}{*}{rank}  & \multirow{2}{*}{max}  & \multirow{2}{*}{max\,|\,st} & \multirow{2}{*}{max\,|\,qdir} & \multirow{2}{*}{int} & \multirow{2}{*}{int\,|\,st} & \multirow{2}{*}{int\,|\,qdir} & \multirow{2}{*}{MBD} & \multirow{2}{*}{MHRD}
\\
/ True model &&&&&&&&&&
        \\\toprule
CSR  & $\hat{L}(r)$ & 0.622 & 0.254 & 0.587 & 0.582 & 0.081 & 0.209 & 0.233 & 0.127 & 0.161 \\
/ Strauss(250, 0.6, 0.03)           & $\hat{J}(r)$ & 0.586 & 0.381 & 0.587 & 0.617 & 0.549 & 0.671 & 0.666 & 0.598 & 0.534 \\
 \midrule
IPP(linear) & $\hat{L}_{\inhom}(r)^{(i)}$ & 0.340 & 0 & 0.288 & 0.300 & 0 & 0 & 0 & 0 & 0 \\ 
/ Strauss(linear, 0.6, 0.03)        & $\hat{J}_{\inhom}(r)^{(i)}$ & 0.628 & 0.390 & 0.560 & 0.636 & 0.584 & 0.696 & 0.700 & 0.638 & 0.572 \\  \midrule
IPP(linear) & $\hat{L}_{\inhom}(r)^{(ii)}$ & 0.552 & 0.238 & 0.524 & 0.552 & 0.082 & 0.166 & 0.196 & 0.134 & 0.164 \\ 
/ Strauss(linear, 0.6, 0.03)        & $\hat{J}_{\inhom}(r)^{(ii)}$ & 0.522 & 0.382 & 0.496 & 0.522 & 0.564 & 0.668 & 0.680 & 0.622 & 0.558 \\ 
\end{tabular}}}
\end{table}

\section{Data examples}\label{sec:datastudy}

In Section \ref{sec:dataex_particles}, the new global envelope tests are applied to the particle data of Section \ref{sec:motivating_ex_and_current_tech}.
Further, Section \ref{sec:ex_tree_data} illustrates testing a composite hypothesis, and Section \ref{sec:ex_fallen_trees} deals with testing an inhomogeneous point process model and mark independence of marks in a marked point pattern.

\subsection{Intramembranous particles: comparing different global envelopes}\label{sec:dataex_particles}

We continue with the example based on intramembranous particles (see Figure \ref{fig:particles}).
We recall that the diagnostic pointwise envelope suggests the rejection of the Gibbs hard core model due to several distances,
as shown in Figure \ref{fig:particles-env99}.
Figure \ref{fig:particles-envelopetests} shows the outputs of the MAD and rank envelope tests.
The unscaled MAD test could not reject the Gibbs hard core model, but it was rejected by all of the new tests. The rejection was explained by the repulsive behaviour of particles at very short distances and subtle clustering 
with slightly larger distances ($\approx 15$ pixels). The attractive behaviour for the larger distances detected by the diagnostic plot (Figure \ref{fig:particles-env99}) was not proved.
The repulsive behaviour and also the subtle clustering may be explained by the variable size of particles.
There may also be physical reasons for attraction \citep{BaumannEtal1990}, which is also a subject of further research \citep{SchladitzEtal2003}.

For the unscaled and scaled MAD envelope tests, we used only $s=99$ simulations.
However, we also performed the scaled tests with $s=2499$. The test results were the same,
but the envelopes were smoother.
In this case, the scaled MAD envelopes appear as a reasonable ``approximation''
for the rank envelope computed from $s=2499$ simulations.

\begin{figure}[htbp]
\centering
\makebox{
\includegraphics[width=1\textwidth]{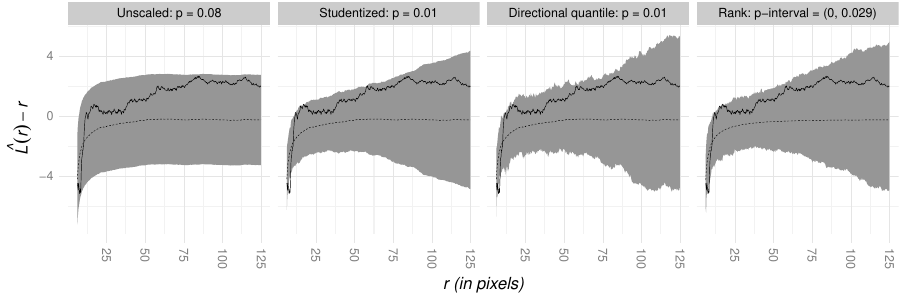}
}
\caption{\label{fig:particles-envelopetests} Three global envelope tests with $T(r)=\hat{L}(r)$ used to test a Gibbs hard core null model for the pattern in Figure \ref{fig:particles}. The unscaled, studentized, and directional quantile MAD envelope tests were based on $s=99$ simulations, and the rank envelope test was based on $s=2499$ simulations. The grey areas show the 95\% global envelopes on $I=[6.1, 125]$ (pixels), the solid black line is the data function and the dashed line represents the (estimated) expectation $T_0(r)$.}
\end{figure}

\subsection{Data example of a tree pattern: testing a composite hypothesis}\label{sec:ex_tree_data} 

The patterns of saplings and large trees in Figure \ref{fig:treedata} originated
from an uneven aged multi-species broadleaf nonmanaged forest in Kaluzhskie Zaseki,
Russia.
These data (all trees) were used by \citet{GrabarnikChiu2002},
who tested the compatibility of the CSR with the data under specific alternative hypotheses.
Subsequently, \citet{Lieshout2010} (in the Handbook of Spatial Statistics) split the data into small and large trees,
and used them to illustrate basic methods for spatial point process statistics.
Similar to \citet{Lieshout2010}, we aimed to quantify the spatial arrangement of the small and large trees.

In the first step, we tested the CSR for the tree patterns in Figure \ref{fig:treedata} using the rank envelope test
with $T(r)=\hat{L}(r)$ (figure omitted). The CSR was clearly rejected for the small trees, but we were not able to reject it for the large trees at the significance level $0.05$. One possible explanation is that the pattern of large trees simply contained too few trees to detect the regularity. Another reason why the CSR hypothesis was not rejected is of biological ground. This is connected with the plasticity of the crowns of trees in broadleaf forests \citep{LonguetaudEtal2013, SchroterEtal2012}. As a consequence, the interactions between trees in broadleaf forests are generally not as strong as those in boreal forests.

\begin{figure}[h!]
\centering\makebox{\includegraphics[width=0.5\textwidth]{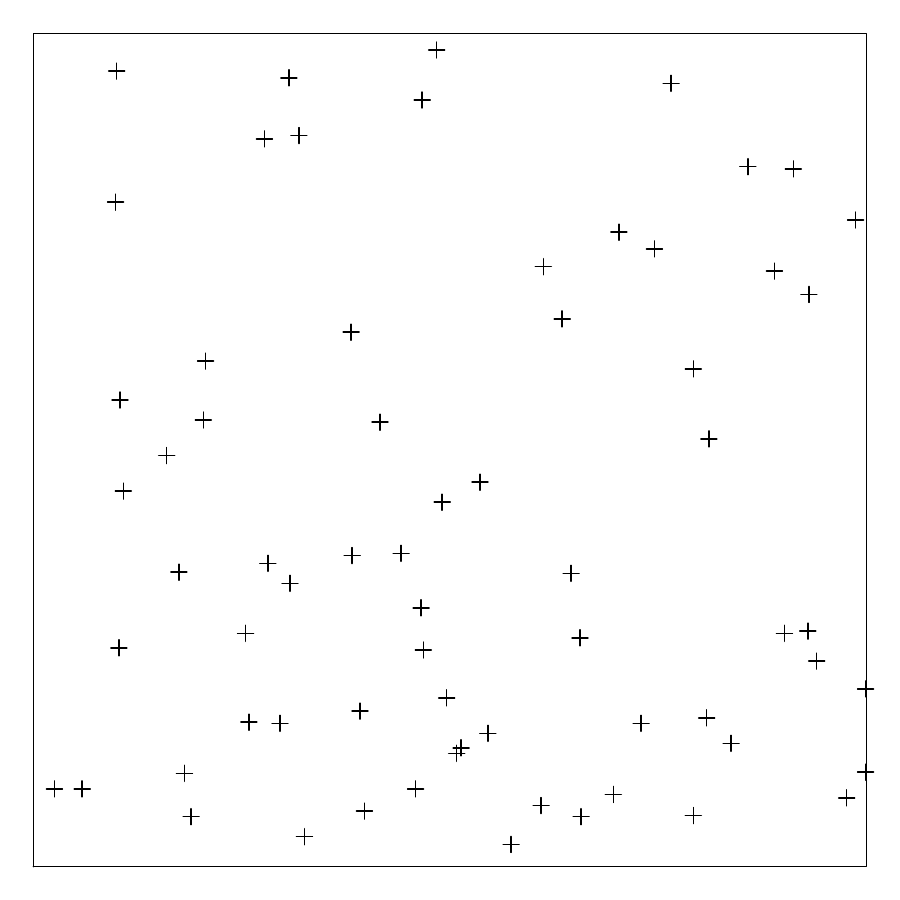}\includegraphics[width=0.5\textwidth]{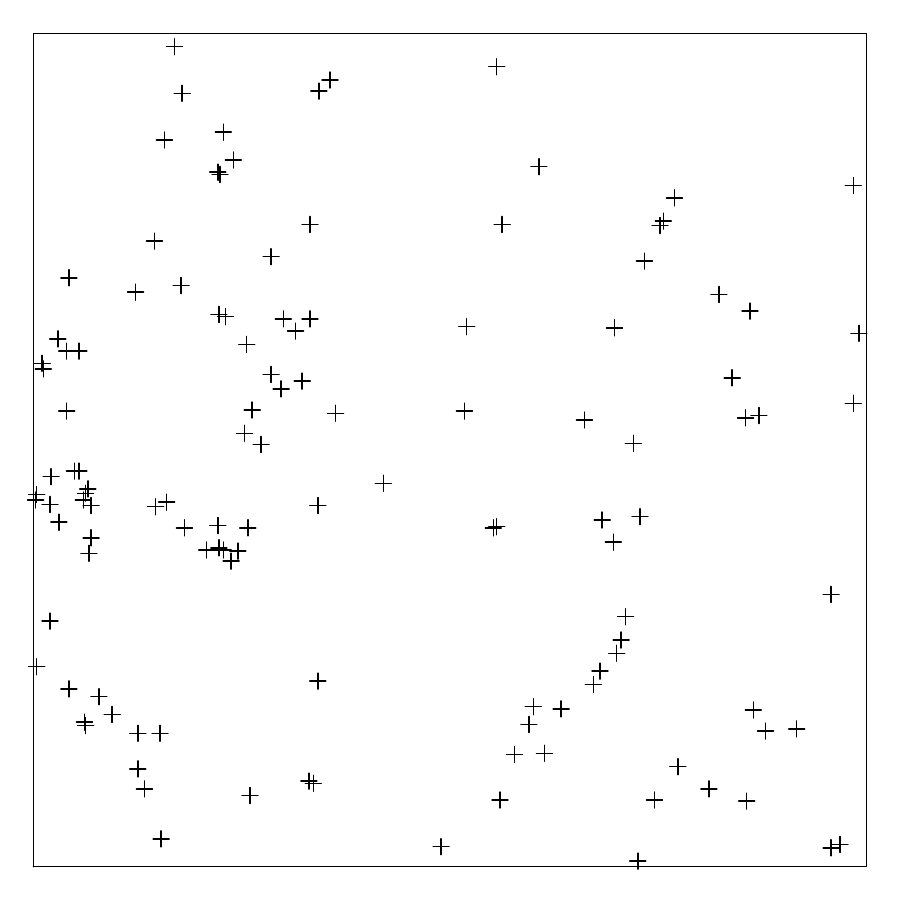}}
\caption{\label{fig:treedata}Locations of trees (0.1 m accuracy in coordinates) in an area of 75 m $\times$ 75 m. Left: 67 trees measuring $>25$ m in height; right: 123 trees measuring $\leq 15$ m in height.}
\end{figure}

The pattern of the small trees was clearly clustered, so we fitted a Mat{\'e}rn cluster process to this pattern.
The estimated parameter values were $\hat{\lambda}_b=0.010$, $\hat{R}_d=3.08$ and $\hat{\mu}_d=2.09$.
Next, we generated $s=2499$ simulations from the fitted model,
and we calculated the conservative (plug-in) rank envelope and the adjusted rank envelope (see Section \ref{sec:composite}).
Whereas the unadjusted critical rank was $k_\alpha = 5$ for $\alpha=0.05$, we obtained the adjusted critical rank $k_\alpha^* = 30$.
Both envelopes are shown in Figure \ref{fig:treedata-saplings-matclust-DG} (left).
The data function was completely inside the adjusted rank envelope, so
the Mat{\'e}rn cluster process was considered to be a suitable model for describing the spatial arrangement of saplings
using this test (the scaled MAD envelope tests based on $s=2499$ gave a very similar result).

We approximated the $\alpha$-quantile from $2499$ additional $p_N$-values of rank counts
based on a lower number of simulations $s_2=99$ and obtained $\alpha^* = 0.25$ and $k_\alpha^* = 35$, which gave a reasonable approximation of the adjusted rank envelope. 

To further explore the goodness-of-fit of the Mat{\'e}rn cluster process we also tested the model with $\hat{J}(r)$, as shown in Figure \ref{fig:treedata-saplings-matclust-DG} (right). In this case, it was interesting that the plug-in and adjusted rank envelopes coincided, which agreed with the simulation study results in Section \ref{sec:simstudy_composite_hypothesis}. Furthermore, this test indicated that the Mat{\'e}rn cluster process was not adequate for describing the pattern of saplings because at distances of $\approx 4$ m, the data pattern appeared to be less clustered than the fitted model.
The Mat{\'e}rn cluster model might be interpreted as a regeneration process in circular gaps between large trees. However, it is possible that the gap openings in the forest were not exactly circular, thereby leading to the rejection of the model by the $J$-function.

\begin{figure}[htbp]
\centering
\makebox{\includegraphics[width=0.5\textwidth]{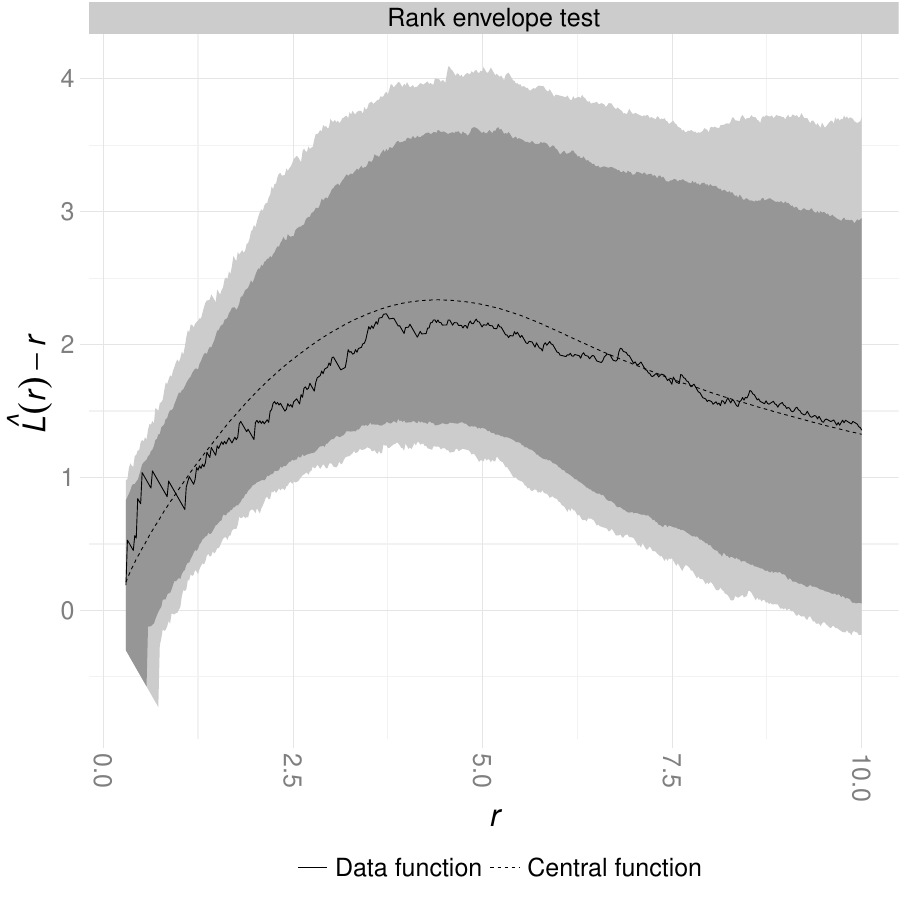}\includegraphics[width=0.5\textwidth]{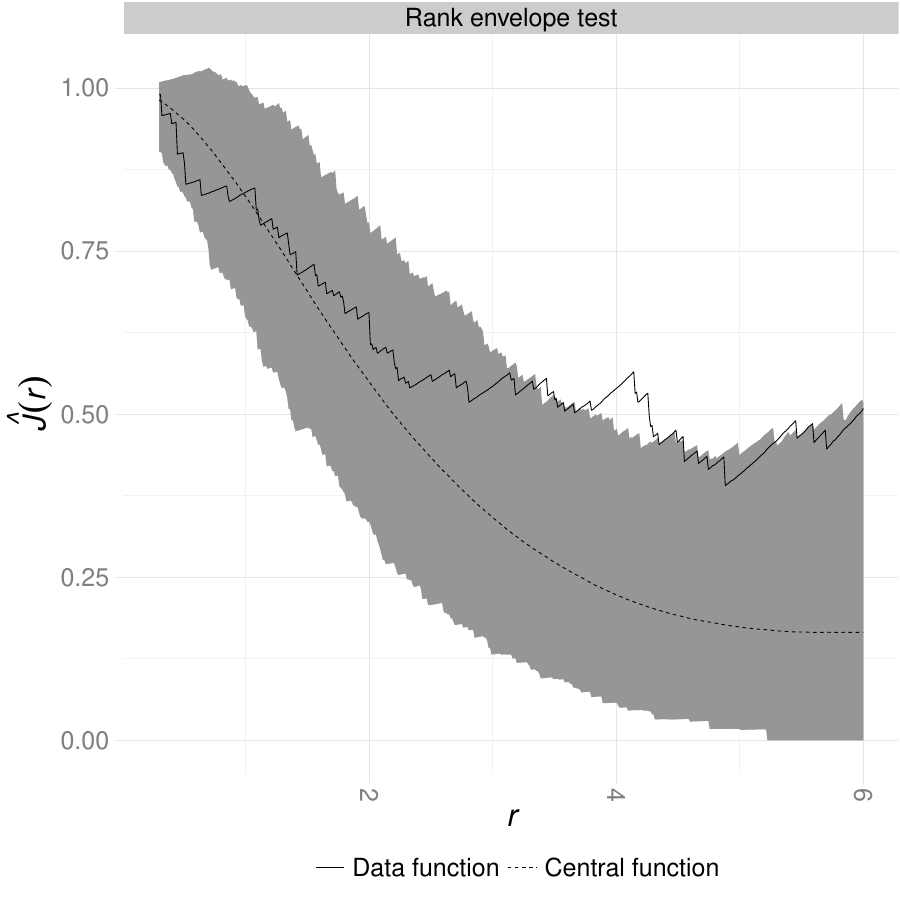}}
\caption{\label{fig:treedata-saplings-matclust-DG} Testing the fitted Mat{\'e}rn cluster process with $T(r)=\hat{L}(r)$ (left) and $T(r)=\hat{J}(r)$ (right) using $s=2499$ for the small trees shown in Figure \ref{fig:treedata} (right). The light and dark grey areas show the convervative (plug-in) and 95\% adjusted rank envelopes on $I=[0.3,10]$, respectively, which both coincide for the $J$-function. The solid black line is the data function and the dashed line represents the (estimated) expectation $T_0(r)$.}
\end{figure}

\subsection{Data example of fallen trees: testing an inhomogeneous model and mark independence}\label{sec:ex_fallen_trees}

In this data example, we show that the testing procedures are not limited to homogeneous point processes, rectangular windows or goodness-of-fit tests for point patterns. In this case, we studied inhomogeneous Poisson process and mark independence hypotheses for the marked point pattern shown in Figure~\ref{fig:trees}.
These hypotheses were studied previously by \cite{MrkvickaEtal2012}, but only the conventional envelope procedures were used.

\begin{figure}[h!]
\centering
\centerline{ \includegraphics[width=0.75\textwidth]{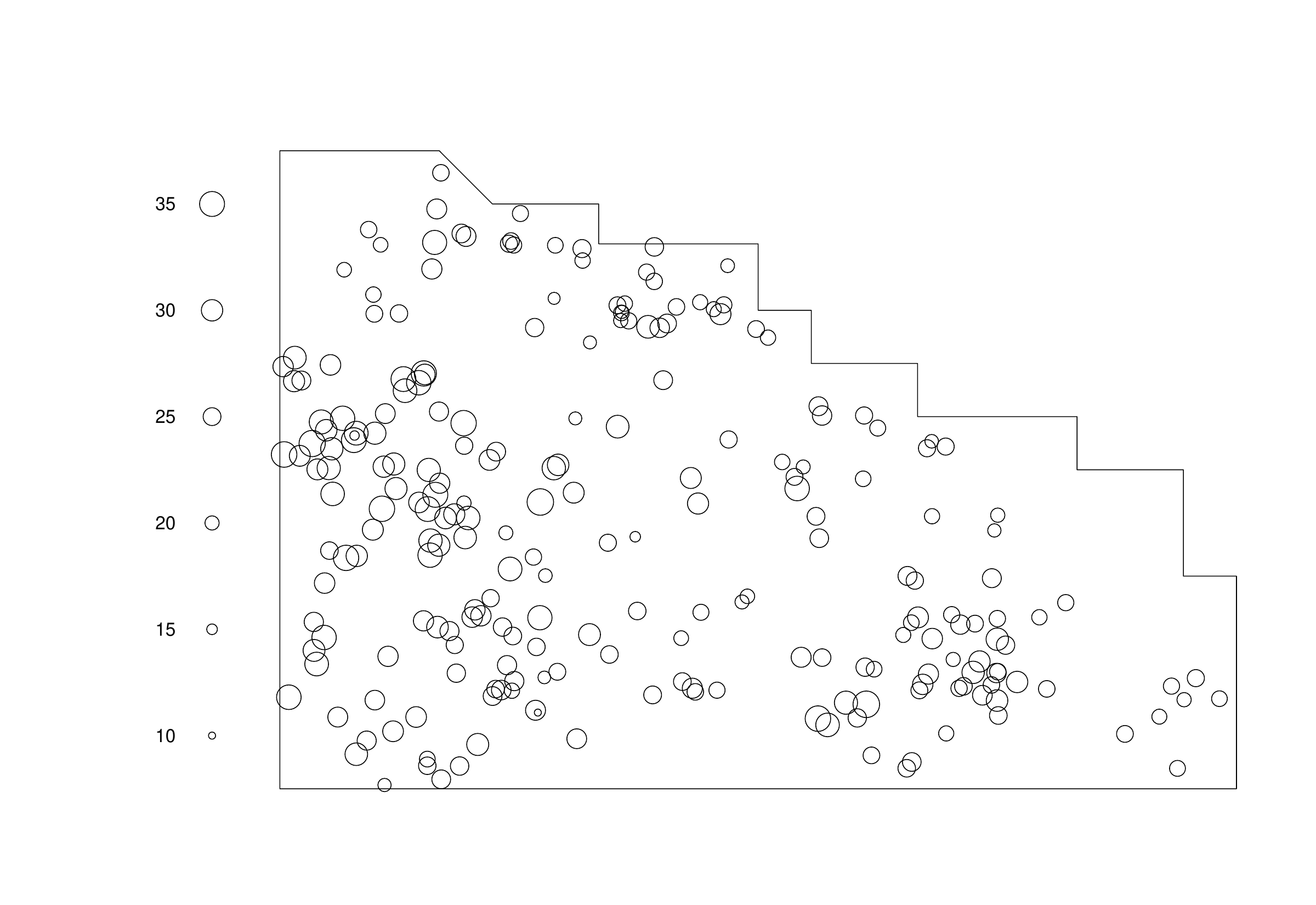}}
\caption{\label{fig:trees}Locations of 231 fallen trees. The longest side lengths of the area are 720 m and 480 m. The sizes of the circles are proportional to the heights of the trees (in meters), as shown in the legend.}
\end{figure}

The dataset comprised the locations and heights of
232 trees, which fell during two large wind gusts (1967 and 1990) in the west of France \citep{PontaillerEtal1997}.
The study area was a biological reserve, which had been preserved for at least four centuries,
with little human influence for a long period \citep{Guinier1950}.
Thus, the forest stand followed almost natural dynamics. It was an uneven-aged beech
stand with a few old oaks.

First, we investigated whether the locations of the fallen trees could be modeled by an inhomogeneous Poisson process with a given intensity function \citep[see e.g.][]{IllianEtal2008}.
The inhomogeneous intensity was estimated by a kernel estimator with a given bandwidth of 53 m,
as described by \citet{MrkvickaEtal2012}, and we used a standard estimator of the inhomogeneous $L$-function \citep{BaddeleyEtal2000} as the functional test statistic.
Our results (Section \ref{sec:simstudy_inhomog}) showed that the estimation of the intensity may affect the type I error probability of the test, so we calculated the adjusted envelope. We chose the option where the intensity was estimated from the data and maintained fixed over all simulations for easier computation (see case (i) in Section \ref{sec:simstudy_inhomog}).

We used the directional quantile envelope test with $s=999$ (the rank envelope test would be too time consuming).
We determined the adjusted level $\alpha^* = 0.084$ and the test is shown in Figure \ref{fig:trees_inhomogPoisson_test} (left).
Thus, the inhomogeneous Poisson process assumption was not rejected by the inhomogeneous $L$-function.
However, the corresponding test using the inhomogeneous $J$-function \citep{Lieshout2011}
indicated the increased clustering of points compared with the inhomogeneous Poisson process using the given bandwidth of 53 m, as shown in Figure \ref{fig:trees_inhomogPoisson_test} (right).
It is possible that the trees that fell during the gusts usually produced gap openings with variable sizes,
which was not captured completely by the inhomogeneous Poisson process where the inhomogeneity
was estimated by a kernel with a fixed width.

\begin{figure}[h!]
\centering
\centerline{ \includegraphics[width=0.5\textwidth]{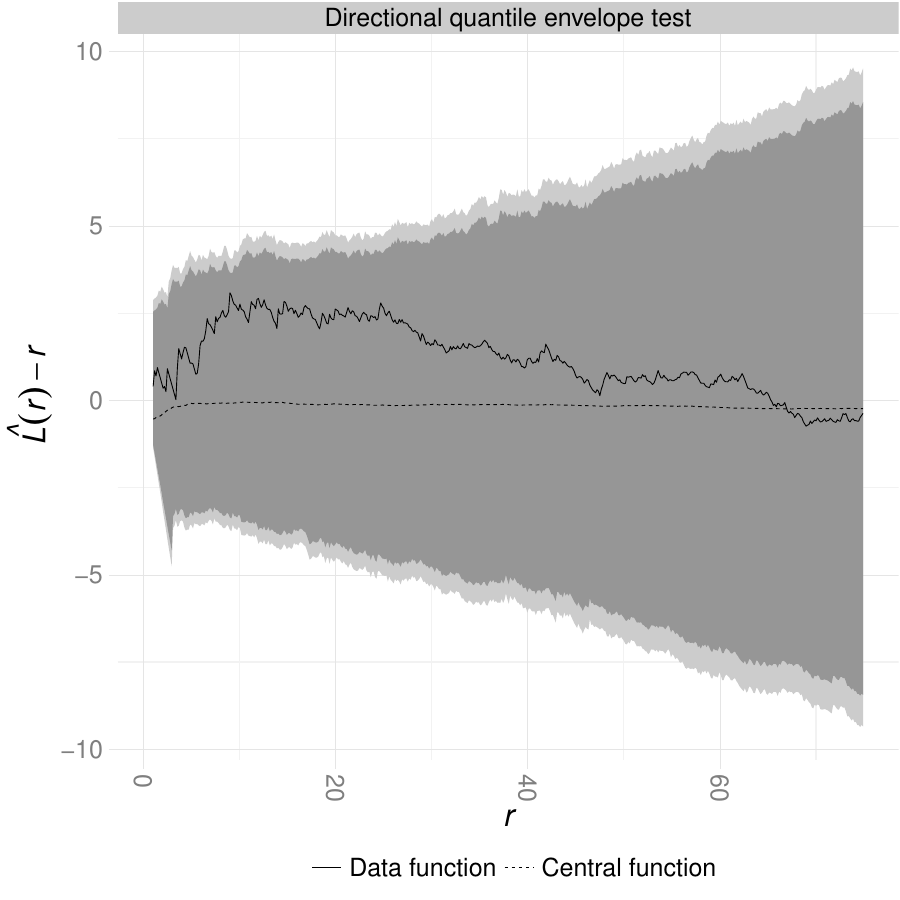}\includegraphics[width=0.5\textwidth]{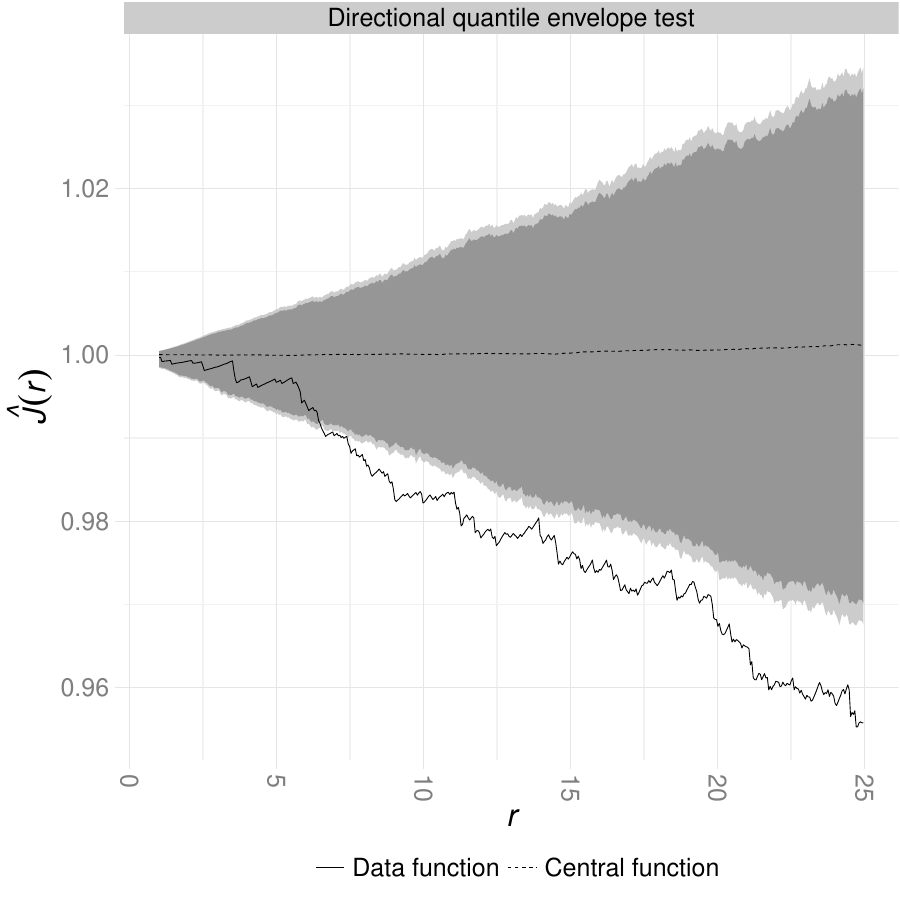}}
\caption{\label{fig:trees_inhomogPoisson_test} Testing the fitted inhomogeneous Poisson process for the trees in Figure \ref{fig:trees} using the inhomogeneous $\hat{L}(r)$ on $I=[1,75]$ (left) and $\hat{J}(r)$ on $I=[1,25]$ (right). The light and dark grey areas show the convervative (plug-in) and 95\% adjusted directional quantile envelopes, respectively, based on $s=999$ simulations. The solid black line denotes the data function and the dashed line represents the (estimated) expectation $T_0(r)$.}
\end{figure}

We also tested the random labelling hypothesis, which states that the marks (heights) in the point pattern are independent of each other. Simulations based on this simple hypothesis were generated by permuting the marks. We used a non-edge corrected estimator of the mark-weighted $L_{\gamma}(r)$-function \citep[see e.g.][]{IllianEtal2008}. Note that the edge correction is not needed, because the function was estimated in a similar manner for the observed and simulated patterns and the non-parametric rank envelope test was used.
The output of the resulting rank envelope test is shown in Figure \ref{fig:treesIndependence}.
The marks were clearly dependent and the dependency was identified for distances ranging from 30 to 100 m (note the cumulative function).
Thus, we can expect trees of similar heights within these distances,
which suggests that soil or other environmental conditions vary in the study area.

\begin{figure}[h!]
\centering
\vspace{0.5cm}
\centerline{ \includegraphics[width=0.5\textwidth]{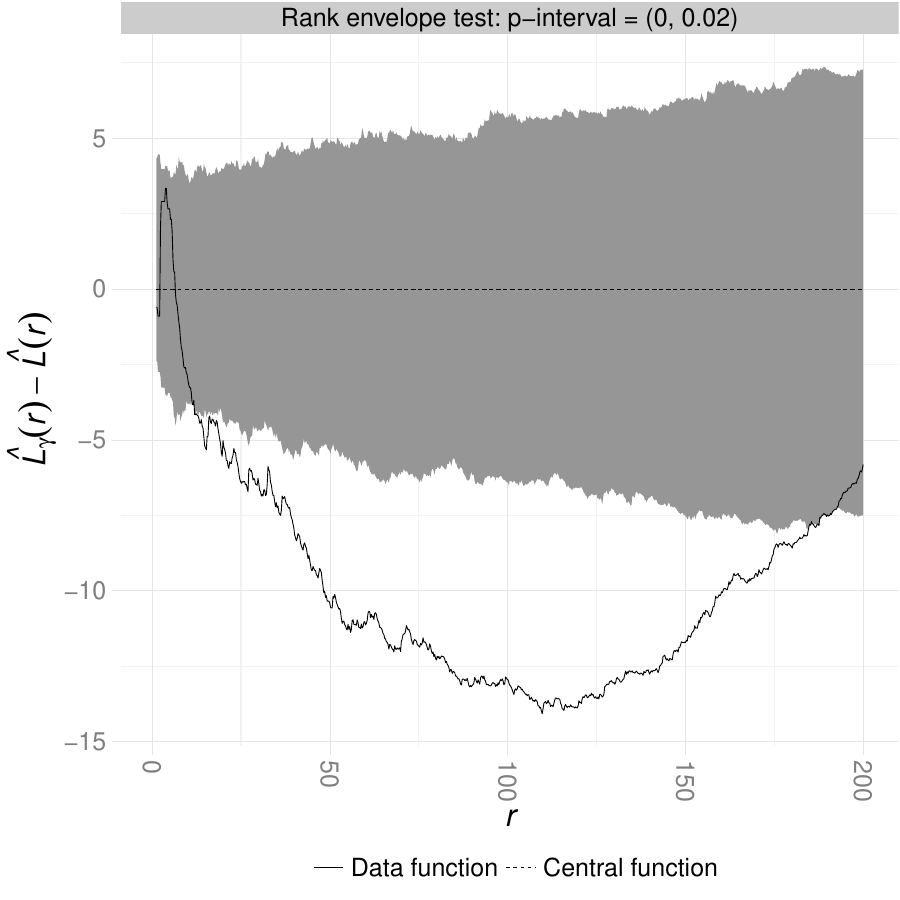}}
\caption{\label{fig:treesIndependence} The rank envelope test for the independence of heights in the pattern shown in Figure \ref{fig:trees} using the centred mark-weighted $\hat{L}_\gamma(r)$ function and $s=2499$ simulations. The grey area shows the 95\% global rank envelope on $I=[0,200]$, the solid black line denotes the data function and the dashed line represents the expectation $T_0(r)$.}
\end{figure}

\section{Discussion and conclusions}\label{sec:discussion}

In this study, we presented new global envelope tests that provide $p$-values and a graphical representation. The $100\cdot(1-\alpha)$\% global envelope complements the test result given by the $p$-values. It shows the distances $r$ where the behaviour of the data function $T_1(r)$ leads to the rejection of the null hypothesis, which helps to understand the reasons of rejection and to suggest alternative models.

In particular, the rank envelope test gives a theoretical basis for
the ``envelope test'' based on
the $k$-th lower and upper curves \eqref{kth_envelopes}.
The key idea is to shift the focus from the values of $T_i(r)$ to the functions
$T_i(r)$ on $r\in I$ and to order the functions by the extreme ranks $R_i$ (see \eqref{rank_measure}).
This allows $p$-values to be calculated in a similar manner to $p$-values in Monte Carlo tests based on scalar discrete statistics. Moreover, the value of $k$ that yields a global envelope of the desired $\alpha$ level can be found easily.

All of the new envelope tests have the desired $\alpha$ level for simple hypotheses.
Furthermore, we proposed a method for constructing adjusted envelopes for composite hypotheses
based on \citet{DaoGenton2014} such that the correct level is also obtained in this case.
A drawback of this adjustment is its extremely high computational load, which may make it infeasible for null models that are heavy to simulate.

The aim of finding the appropriate boundaries $T_\low$ and $T_\upp$ of an envelope is
related to a classical problem in stochastic processes theory, i.e.,
computing the probability that a random process remains between two given boundaries.
This difficult question has been a research focus for many years.
However, progress has been made only in a few cases, e.g.,
for Brownian motion \citep{Durbin1971, Durbin1992, Lerche1986, Daniels1996}
and Gaussian processes \citep{CramerLeadbetter1967, LeadbetterEtal1983, Berman1992, AdlerTaylor2009}.
In the present study, the random functions $T_i$ do not belong to these cases, and thus we resort to Monte Carlo simulation.
This approach allows us to work with any spatial process models,
although the cost of universality may be a high computational load.

The rank envelope test can be recommended if a large
number of simulations $s$ can be performed, where for $\alpha=0.05$, the appropriate number of simulations
$s$ is close to 2500. This is a completely non-parametric
test, which by construction, gives the same importance to all distances $r$
on the chosen interval $I$, and this is a natural choice
if there is no {\em a priori} single interesting distance $r_0$.
Our simulation study also showed that the rank envelope test performs better
than the classical deviation tests, which are the state-of-the-art methods in testing for spatial point patterns.
Most importantly, the rank envelope test is extremely robust against the choice of $I$, unlike the classical deviation tests. Our simulation study (see Figure \ref{fig:choiceL}) suggests that it is in principle possible to choose $I$ wisely such that the classical tests have reasonable power in comparison to the new tests. However, it is typically difficult to make this choice {\em a priori}. For the rank envelope test, the interval $I$ can be selected to be large enough to cover all possibly interesting distances without the risk of losing the power of the test.

The scaled MAD envelope tests can be used with a low number of simulations.
As shown in this study, scaled MAD envelopes based on $s=99$ can give
a reasonable approximation for the rank envelope computed from $s=2499$ simulations.
However, this behaviour is not guaranteed in all cases because the studentized
and directional quantile envelopes rely only on one or two characteristics of the distribution
of $T(r)$, $r\in I$, i.e., the variance or quantiles.
Regardless, our examples show that if a large $s$ is not possible, then
these tests can be considered as alternatives.
However, it should be noted
that many more simulations are required to obtain precise estimates
of the null distribution of a test statistic using
the classical Monte Carlo approach.
In particular, for estimated $p$-values as low as $\hat{p}<0.1$, the
acceptable accuracy requires $s=1000$ or more simulations \citep[see e.g.][]{LoosmoreFord2006}.

To utilize the completely non-parametric rank envelope test
but without a large $s$, a sequential scheme similar to that proposed
by \citet{BesagClifford1991}
can help to make testing computationally cheaper in terms of the
number of simulations $s$ when the data suggest that there is
no evidence for rejecting $H_0$.
Moreover, the rank count test (Section \ref{sec:goldfirst}), which can be seen as a method for breaking ties in the rank envelope test,
can also be used with a low number of simulations $s$, although the envelope representation is then lost. Furthermore, the power of the test with a low $s$ may be lower than that with a high $s$ since the number of simulations $s$ controls the nature of the test, i.e., for small values of $s$, the rank count measure is integral type, whereas it is maximum type for large values of $s$.

Obviously, there are other non-parametric methods for ordering the functions $T_i$
in addition to the extreme ranks $R_i$ or the rank counts $\Nvec_i$.
Thus, we included two orderings (MBD and MHRD) from the area of functional data analysis in our study.
However, these orderings had very low rejection rates for our null and alternative models.
We also note that the construction of $100\cdot(1-\alpha)$\% global envelopes is not straightforward
for these integral type tests.

Intuitively, it is clear that the power of the test depends on the choice of
the function $T(r)$. Sometimes, it may be possible to suggest a sensible
function based on the null and alternative models. However, in the case
where the preferred test function is unclear, it may be desirable to base the
test on several test functions \citep[see][]{Mrkvicka2009}.
Performing the same test with many functions leads to another multiple
testing problem, which can also be solved using the idea of the rank envelope test.
This multiple testing adjustment will be a topic of future research.

\section*{Acknowledgements}
We are grateful to Dietrich Stoyan (TU Bergakademie Freiberg) for his valuable comments and suggestions regarding an earlier version of this paper. We also thank Torsten Mattfeldt (Ulm University) for providing the intramembranous particle data, as well as two anonymous reviewers and editors for useful comments that led to improvements in the paper. M.\,M.\ was financially supported by the Academy of Finland (project numbers 250860, 294162), T.\,M.\ by the Grant Agency of the Czech Republic (Projects No.\ P201/10/0472), P.\,G.\ by an RFBR grant (project number 12-04-01527) and U.\,H. by the Centre for Stochastic Geometry and Advanced Bioimaging, funded by a grant from the Villum Foundation.
The work was started when M.\,M.\ was at Aalto University, Finland.

\appendix

\section*{Appendix}

\section{Functional test statistics}\label{sec:summary_functions}

As functional test statistics, we used non-parametric estimators of the $L$-function
\citep{Ripley1976, Ripley1977, Besag1977} and the $J$-function
\citep{LieshoutBaddeley1996}. Both functions are often used for
detecting the clustering or regularity of point patterns.

Recall that the $J$-function is defined as
\begin{equation*}
 J(r) = \frac{1-G(r)}{1-F(r)} \quad\text{for } r\geq 0 \text{ with } F(r)<1,
\end{equation*}
where $G$ is the nearest neighbour distribution function and $F$ is
the empty space function (or spherical contact distribution
function). Under CSR, it holds that $J(r) \equiv 1$.
As shown by \citet{BaddeleyEtal2000}, it holds that $(1-\hat{G}(r))/(1-\hat{F}(r))
\approx 1$ for the CSR case also if uncorrected estimators of $G(r)$
and $F(r)$ are used. Consequently, we used the uncorrected estimator
of $J(r)$, which is provided in the R library spatstat
\citep{BaddeleyTurner2005}.
Otherwise, for the $G$- and $F$-functions (in Appendix \ref{sec:s}), Kaplan-Meier estimators were used \citep{BaddeleyGill1997}.
For the $L$-function, we employed an estimator with translational edge correction \citep[see e.g.][]{IllianEtal2008}.

\section{Number of simulations in the rank envelope test}\label{sec:s}

The test decision is unambiguous at level $\alpha$ if either $p_+ \leq \alpha$ (clear rejection) or $p_- > \alpha$ (no evidence against $H_0$). The case where $\alpha \in [p_-, p_+)$ is considered to be unsatisfactory. However, it occurs only rarely if the width of the interval,
\begin{equation}\label{p-interval-width}
p_+ - p_- = \frac{1}{s+1}\sum_{i=1}^{s+1} \1(R_i = R_1),
\end{equation}
is small.
This width depends on the number of simulations $s$, the smoothness properties of $T(r)$ and
the value of $R_1$. We investigated the features of these dependencies using simulations.

We observed (Figure \ref{fig:width2}) that
the width of the $p$-interval \eqref{p-interval-width} is largest for $R_1 = 1$.
In order to show how this maximal width varies with respect to the choice of
$T(r)$ and the null model, we recorded the average maximal width
(computed from 20 repetitions) for various functions $T(r)$ and null models with
respect to the number of simulations $s$.
As $T(r)$, we selected non-parametric estimators of the classical summary functions for
point processes ($L$, $F$, $G$, $J$), and as null models, we used Poisson,
Mat\'ern cluster and Strauss point process models
(the simulated models were Poisson(200), MatClust(200, 0.06, 4) and Strauss(350, 0.4, 0.03); see Section \ref{sec:ppmodels}).
The results are summarized in Figure \ref{fig:width}.
The maximal width is clearly of order $s^{-1}$ for all of the functions and models that we studied
(all of the slopes are insignificant).
In fact, the maximal width of the p-interval is generally in the order of $s^{-1}$:
The curves are always evaluated only at finitely many points $r\in I_\fin$, so the number of curves with a given extreme rank cannot be larger than $2|I_\fin|$, where $|I_\fin|$ denotes the number of distances. Therefore, 
\begin{equation}\label{eq:order_s_result}
  \frac{\sum_{i=1}^{s+1}\1(R_i=R_1)}{s+1}\leq 2|I_\fin|/(s+1).
\end{equation}

\begin{center}
\begin{figure}[htbp]
\includegraphics[width=4.5cm,angle=0]{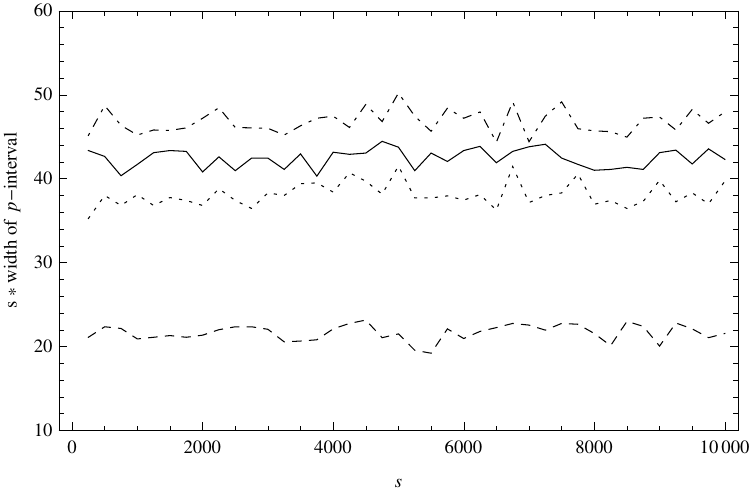}
\includegraphics[width=4.5cm,angle=0]{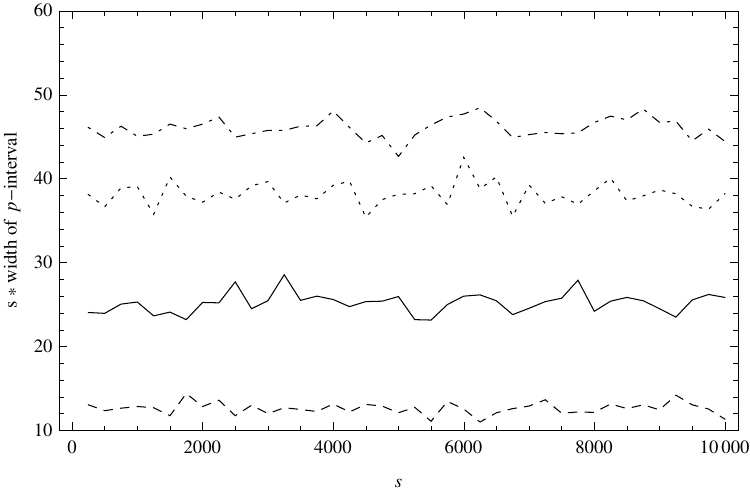}
\includegraphics[width=4.5cm,angle=0]{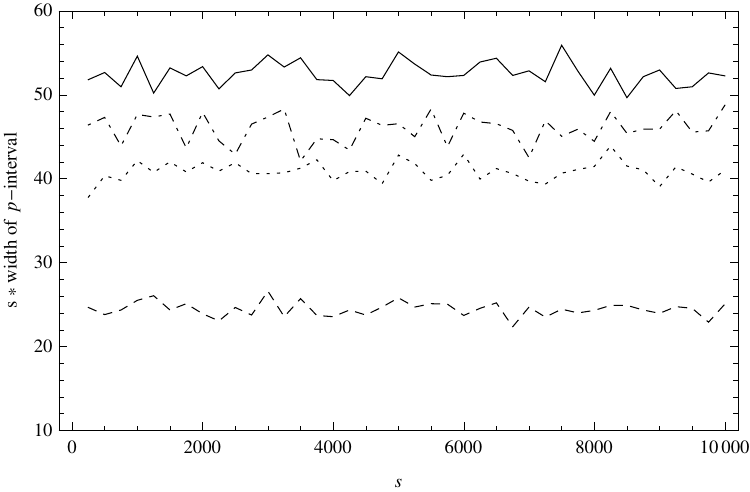}
\caption{\label{fig:width}Average maximal width of the $p$-interval multiplied by the number of simulations $s$ with respect to $s$ for: $L$- (solid line), $F$- (dashed line), $G$- (dotted line) and $J$-functions (dashed dotted line). {\em Left}: Poisson null model; {\em middle}: Mat\'ern cluster model; {\em right}: Strauss model.}
\end{figure}
\end{center}

For the common choice $\alpha=0.05$, $\alpha\in [p_-,p_+)$ only ``rarely''
if the width of the $p$-interval is less than or equal to 0.01.
Since all of the functions shown in Figure \ref{fig:width} are bounded approximately by 50,
the width $0.01$ for $R_1=1$ was achieved for $s=5000$ simulations in the cases considered.

In fact, the width only needs to be narrow around the significance level $\alpha$.
Figure \ref{fig:width2} shows the widths \eqref{p-interval-width} for
$R_1=1, 2, 3$ and $4$ for the Poisson null model using the $L$- and $J$-functions
(we also checked them for the other functions and null models).
With $s=2500$, the $p$-intervals corresponding to values $R_1=1, 2, 3, 4$ are approximately $(0, 0.02)$, $(0.02, 0.032)$, $(0.032, 0.042)$ and $(0.042, 0.052)$,
so the values $R_1=3$ or $4$ correspond to $p$-values close to $\alpha=0.05$.
Therefore, we can allow the width to be larger than 0.01 for $R_1=1$ (and 2) and
require the precision 0.01 only for $R_1=3$ and $4$.
We conclude that, for $\alpha = 0.05$, $s\approx 2500$ simulations appear to be sufficiently many.

We also found that the maximal width related to the $F$-function, which tends to be rather smooth, is smaller than those related to the other functions considered, as shown in Figure \ref{fig:width}. 
The differences in the maximal widths are small between the three different null models, except 
the $L$- and $F$-functions appear to be smoother with the Mat\'ern cluster model than those with the other two models resulting in smaller maximal widths.

\begin{figure}[htbp]
\centering
\includegraphics[width=0.5\textwidth]{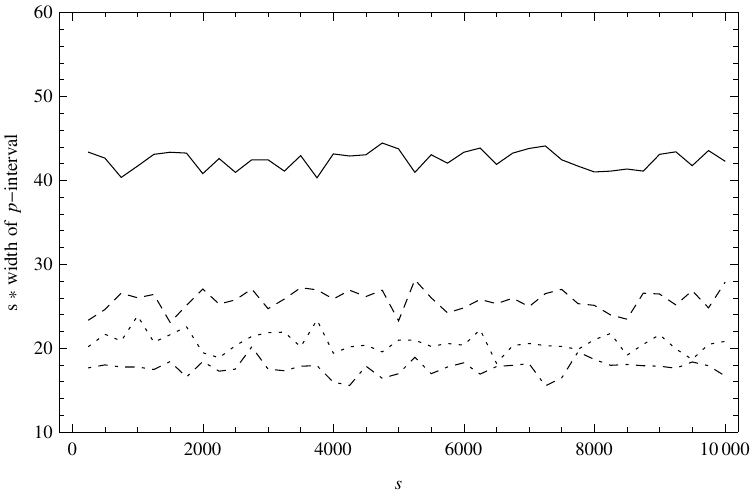}\includegraphics[width=0.5\textwidth]{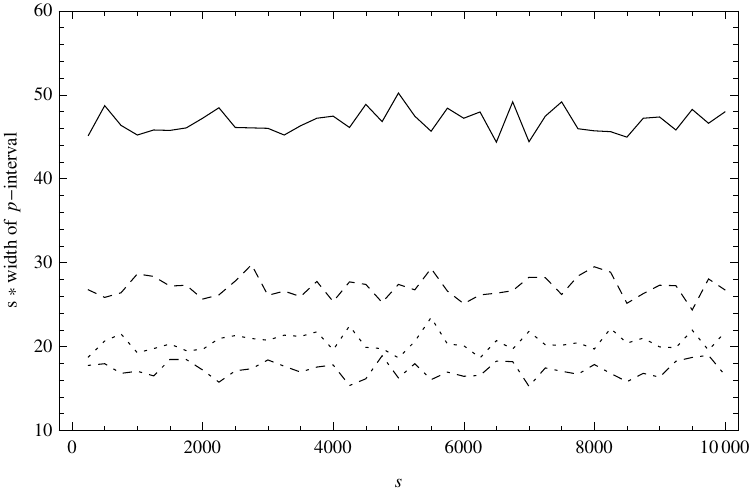}
\caption{\label{fig:width2}Average width of the $p$-interval multiplied by the number of simulations $s$ with respect to $s$ for the Poisson null models for $R_1=1$ (solid line), 2 (dashed line), 3 (dotted line) and 4 (dashed dotted line). {\em Left}: $L$-function; {\em right}: $J$-function.}
\end{figure}

The width \eqref{p-interval-width} also depends on the number of $r$-values;
thus, Figure \ref{fig:width3} shows that for all of the functional statistics considered,
it behaves logarithmically with respect to the number of $r$-values, $|I_\fin|$.
Furthermore, the experiments not shown here indicate that the width grows as the width of the interval $I$ increases. 
The selected intervals cover the overall common interaction domain.
In any case, we recommend to choose the interval $I$ so that it covers the domain of all potential interactions (see results in Section \ref{sec:simstudy_I}) and we recommend to take the number of $r$-values sufficiently large so that behaviour of functions on $I$ are captured adequately.  
These choices guarantee that the test can have good power against many alternative models.

Given all the reasons above, we recommend the use of at least $s=2500$ simulations for testing at the significance level  $\alpha=0.05$.

\begin{center}
\begin{figure}[htbp]
\centering
\includegraphics[width=0.5\textwidth]{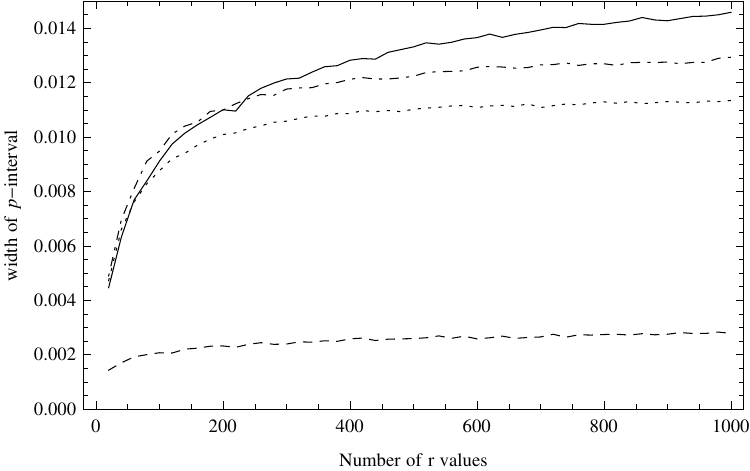}
\caption{\label{fig:width3}Average maximal width of the $p$-interval with respect to the selected number of $r$ values for the Poisson null model for $L$- (solid line), $F$- (dashed line), $G$- (dotted line) and $J$-functions (dashed dotted line).}
\end{figure}
\end{center}

\section{Proof of Theorem \ref{thm:gf-envelope}}\label{sec:proof_thm6.2} 

For the proof, we use the relations between the extreme ranks and rank counts, which follow from the definition.
\begin{align}
  R_i < R_{i^\prime} &\implies \Nvec_i \prec \Nvec_{i^\prime}\label{eq:extreme less imply goldfirst}\\
  \Nvec_i = \Nvec_{i^\prime} &\implies R_i = R_{i^\prime}\label{eq:goldfirst eq imply extreme}
\end{align}
To prove \eqref{eq:env smaller gf}, 
we show that
\begin{equation}\label{eq:env implies gf}
  \varphi_{\env,\cons}(T_1)=1\implies \varphi_N(T_1)=1
\end{equation}
for all possible arguments $T_1$.

Note first that $\varphi_{\env,\cons}(T_1)=1$ iff the curve $T_1$ leaves the envelope in one of the discrete $r$-values. As stated in Theorem \ref{thm:envelope-vs-pinterval}, 
this is equivalent to $p_+ \leq \alpha$, and with the definition of $p_+$ (see \eqref{eq:pvalue_interval}), 
we have
\begin{align*}
  \varphi_{\env,\cons}(T_1)=1 &\Longleftrightarrow \sum_{i=2}^{s+1}\1(R_i > R_1) \geq (s+1)(1-\alpha)
  \intertext{and, by \eqref{eq:extreme less imply goldfirst}, $\sum_{i=2}^{s+1}\1(\Nvec_i \succ \Nvec_1)\geq\sum_{i=2}^{s+1}\1(R_i > R_1)$, thus}
  \varphi_{\env,\cons}(T_1)=1&\implies \sum_{i=2}^{s+1}\1(\Nvec_i\succ\Nvec_1) \geq (s+1)(1-\alpha)
    \\&\Longleftrightarrow \sum_{i=2}^{s+1}\1(\Nvec_i\precsim\Nvec_1) \leq \alpha(s+1)-1
  \\&\Longleftrightarrow \sum_{i=1}^{s+1}\1(\Nvec_i\precsim\Nvec_1) \leq \alpha(s+1)\\&\Longleftrightarrow \varphi_N(T_1)=1.
\end{align*}
In order to prove \eqref{eq:env goesto gf}, 
we can consider the difference $\varphi_N(T_1)-\varphi_{\env,\cons}(T_1)$. For the realizations of $T_1$ where both tests yield the same result, the difference vanishes. We can also see that $\varphi_N(T_1)\geq\varphi_{\env,\cons}(T_1)$ for all $T_1$. We still need to examine the cases where $\varphi_N(T_1)=1$ but $\varphi_{\env,\cons}(T_1)=0$.
This is only possible when $R_1=k_\alpha$: If $R_1<k_\alpha$, then $\varphi_{\env,\cons}(T_1)=1$. On the other hand, if we assume that $R_1>k_\alpha$, then by the definition of $k_\alpha$
and \eqref{eq:extreme less imply goldfirst},
\begin{equation*}
  \sum_{i=2}^{s+1}\1(R_i<R_1) >\alpha(s+1) \ \implies\
  \sum_{i=2}^{s+1}\1(\Nvec_i\prec\Nvec_1) >\alpha(s+1) \ \implies\
 \sum_{i=2}^{s+1}\1(\Nvec_i\precsim\Nvec_1) >\alpha(s+1)
\end{equation*}
which is the same as $\varphi_N(T_1)=0$.
To summarize,
$\varphi_N(T_1)-\varphi_{\env,\cons}(T_1) =1$ only if $R_1=k_\alpha$, and thus
\[
  \E\big(\varphi_N(T_1)-\varphi_{\env,\cons}(T_1))\leq \E\big(\1(R_1=k_\alpha)\big)=\Pr(R_1=k_\alpha).
\]

Monte Carlo goodness-of-fit tests exploit the fact that for any measurable event $A$,
\[
  \Pr(T_1,\dots,T_{s+1}\in A \mid \{T_1,\dots,T_{s+1}\})=\Pr(T_{\sigma(1)},\dots,T_{\sigma(s+1)}\in A\mid \{T_1,\dots,T_{s+1}\})
\]
for any permutation $\sigma$ and any measurable $A$, particularly for events $A$ that concern rankings of $T_i$, $i=1,\dots, s+1$, such as $A=\{R_1=k_\alpha\}$.
As a consequence of equiprobability for the permutations,
\[
\Pr(R_1=k_\alpha\mid T_1=t_1,\dots,T_{s+1}=t_{s+1}) = \frac{\sum_{i=1}^{s+1}\1(R_i=k_\alpha)}{s+1}.
\]
The rest follows from Eq.\ \eqref{eq:order_s_result}, which yields the proposed limit.

\section{Adjusted global rank envelope for a composite hypothesis}\label{sec:adjenv}

The adjusted global rank envelope for a composite hypothesis can be constructed as follows.
\begin{enumerate}
 \item Estimate the parameters $\thetavec$ from the data pattern $x_1$.
 \item Simulate the null model $s$ times with the estimated parameters $\hat{\thetavec}_1$.
Let $x_2(\hat{\thetavec}_1), \dots, x_{s+1}(\hat{\thetavec}_1)$ be the simulated patterns. Calculate the functional test statistic $T(r)$ for $x_1,x_2(\hat{\thetavec}_1), \dots, x_{s+1}(\hat{\thetavec}_1)$.
 \item Calculate the extreme rank $R_1$ based on the ranking of the functional statistics $T_1(r),\dots,T_{s+1}(r)$.
 \item Run steps (i)--(iii) for the $s$ simulated patterns as the data patterns and obtain the plug-in ranks $R_{1,i}, i=2,\dots,s+1$, as follows.
  \begin{itemize}
   \item Estimate the parameters $\thetavec$ from the pattern $x_i(\hat{\thetavec}_1)$.
   \item Simulate the null model $s$ times with the estimated parameters $\hat{\thetavec}_i$. Calculate the function $T(r)$ for $x_i(\hat{\thetavec}_1)$ and the new simulated patterns.
   \item Calculate the extreme rank $R_1$ for $x_i(\hat{\thetavec}_1)$ and denote it as $R_{1,i}$.
  \end{itemize}
 \item Determine the $\alpha$-quantile, i.e.,\ the critical rank $k^*_\alpha$, from the sample $(R_{1,2},..., R_{1,s+1})$,
\begin{equation*}
  k_\alpha^* = \max\left\{k:\  \sum_{i=2}^{s+1}\1(R_{1,i} < k) \leq \alpha s\right\}.
\end{equation*}
	
 \item The adjusted rank envelope is given by the curves $T_\low^{(k_\alpha^*)}$ and $T_\upp^{(k_\alpha^*)}$ (see \eqref{kth_envelopes}). 
\end{enumerate}

\end{document}